\documentclass[a4paper,fleqn]{cas-dc}
\renewcommand{\printorcid}{}

\usepackage[sort,numbers]{natbib}
\usepackage{amsmath,amsthm,amsfonts,amssymb,amscd}
\usepackage{lastpage,indentfirst}
\usepackage{enumerate}
\usepackage{subcaption}
\usepackage{fancyhdr}
\usepackage{mathrsfs}
\usepackage{xcolor}
\usepackage{IEEEtrantools}
\usepackage{graphicx}
\usepackage[english]{babel}
\usepackage{hyperref}
\usepackage{algpseudocode}
\usepackage[toc,page]{appendix}

\usepackage{todonotes}
\usepackage{latexsym,amssymb,amsmath, amsbsy,amsopn, amstext}
\usepackage{epsfig,tikz,pgf,hyperref,cancel,color,float,ifpdf}
\usepackage{extarrows,mathrsfs}
\usepackage{amsmath,amssymb,stfloats}
\usepackage{tikz,subcaption}
\usepackage[]{caption}
\usetikzlibrary{automata, positioning, arrows}
\usepackage[ruled,linesnumbered]{algorithm2e}
\usepackage{bm}

\newtheorem{proposition}{\bfseries Proposition}
\newtheorem{example}{\bfseries Example}
\newtheorem{assumption}{\it Assumption}
\newtheorem{theorem}{\bfseries Theorem}
\newtheorem{lemma}{\bfseries Lemma}

\newtheorem{remark}{\bfseries Remark}

\newtheorem{problem}{\bfseries Problem}

\newcommand{\R}{\mathbb{R}}

\newcommand{\F}{\mathscr{F}}

\newcommand{\C}{\mathscr{C}}

\newcommand{\Q}{\mathscr{Q}}
\newcommand{\A}{\mathscr{A}}
\newcommand{\mcp}{\mathscr{P}}
\newcommand{\la}{\lambda}
\newcommand{\di}{{\rm d}}
\newcommand{\ta}{\theta}
\newcommand{\ep}{\epsilon}
\newcommand{\ga}{\gamma}
\newcommand{\pa}{\partial}

\newcommand{\ka}{\kappa}

\def\tsc#1{\csdef{#1}{\textsc{\lowercase{#1}}\xspace}}
\tsc{WGM}
\tsc{QE}

\newif\ifdraft
\drafttrue

\makeatletter
\renewcommand\normalsize{%
\@setfontsize\normalsize\@xpt\@xiipt
\abovedisplayskip 2\p@ \@plus2\p@ \@minus5\p@
\abovedisplayshortskip \z@ \@plus3\p@
\belowdisplayshortskip 2\p@ \@plus3\p@ \@minus3\p@
\belowdisplayskip \abovedisplayskip
\let\@listi\@listI}
\makeatother

\begin{document}
\begin{sloppypar}
\let\WriteBookmarks\relax
\def\floatpagepagefraction{1}
\def\textpagefraction{.001}

\shorttitle{}

\shortauthors{Y. Wang and X. Xu}

\title [mode = title]{Adaptive Safety-Critical Control for a Class of Nonlinear Systems with Parametric Uncertainties: A Control Barrier Function Approach}

\tnotetext[1]{This work was supported in part by National Science Foundation Grant 2209791 and 2222541.}

\author[1]{Yujie Wang}
\ead{yujie.wang@wisc.edu}

\author[1]{Xiangru Xu}
\ead{xiangru.xu@wisc.edu}

\address[1]{Department of Mechanical Engineering, University of Wisconsin-Madison, WI 53706, USA}

\begin{abstract}
This paper presents a novel approach for the safe control design of systems with parametric uncertainties in both drift terms and control-input matrices. The method combines control barrier functions and adaptive laws to generate a safe controller through a nonlinear program with an  explicitly given closed-form solution. The proposed approach verifies the non-emptiness of the admissible control set independently of online parameter estimations, which can ensure that the safe controller is singularity-free. A data-driven algorithm is also developed to improve the performance of the proposed controller by tightening the bounds of the unknown parameters. The effectiveness of the control scheme is demonstrated through numerical simulations.
\end{abstract}

\begin{keywords}
Safety-critical Control \sep Control Barrier Functions \sep Adaptive Control \sep Data-Driven Approach \sep Nonlinear Programming
\end{keywords}

\maketitle

\section{Introduction}
\label{sec:introduction}

Control barrier functions (CBFs) have been recently proposed as a systematic approach to ensure the forward invariance of control-affine systems \cite{Xu2015ADHS,ames2016control}. By including the CBF condition into a convex quadratic program (QP), a CBF-QP-based controller can act as a safety filter that modifies  potentially unsafe control inputs in a minimally invasive fashion. However, most existing CBF works require precise model information, which is often challenging to obtain. Robust CBF control methods have been proposed to address this issue, ensuring safety in the presence of bounded model uncertainties \cite{garg2021robust,nguyen2021robust,verginis2021safety,buch2021robust,wang2022disturbance}. Nevertheless, the design of a robust CBF controller relies on the bounds of the uncertainties or the Lipschitzness of the unknown dynamics, making it difficult to handle parametric uncertainties that are generally unbounded.

Adaptive control aims to achieve stabilization or desired tracking performance for uncertain dynamic systems through an adaptive law, and has been extensively studied in the past decades \cite{narendra1989stable,aastrom1995adaptive,KKK95,ioannou1996robust,astolfi2008nonlinear}. Most adaptive control strategies are based on uncertainty parameterization and the certainty equivalence principle, which means that the estimated parameters are used as if they are the true parameters in the feedback control design. For uncertain nonlinear systems in some canonical forms, many adaptive control design techniques have been developed using feedback linearization \cite{sastry1989adaptive,kanellakopoulos1991systematic}, backstepping \cite{KKK92,KKK95}, or averaging \cite{anderson1986stability,kosut1987stability}. A summary of the fundamental theoretical concepts and technical issues involved in multivariable adaptive control is documented in \cite{tao2014multivariable}, and a historical overview of adaptive control and its intersection with learning is provided in \cite{annaswamy2021historical}.

\begin{figure}[!t]
\centering
  \includegraphics[width=0.25\textwidth]{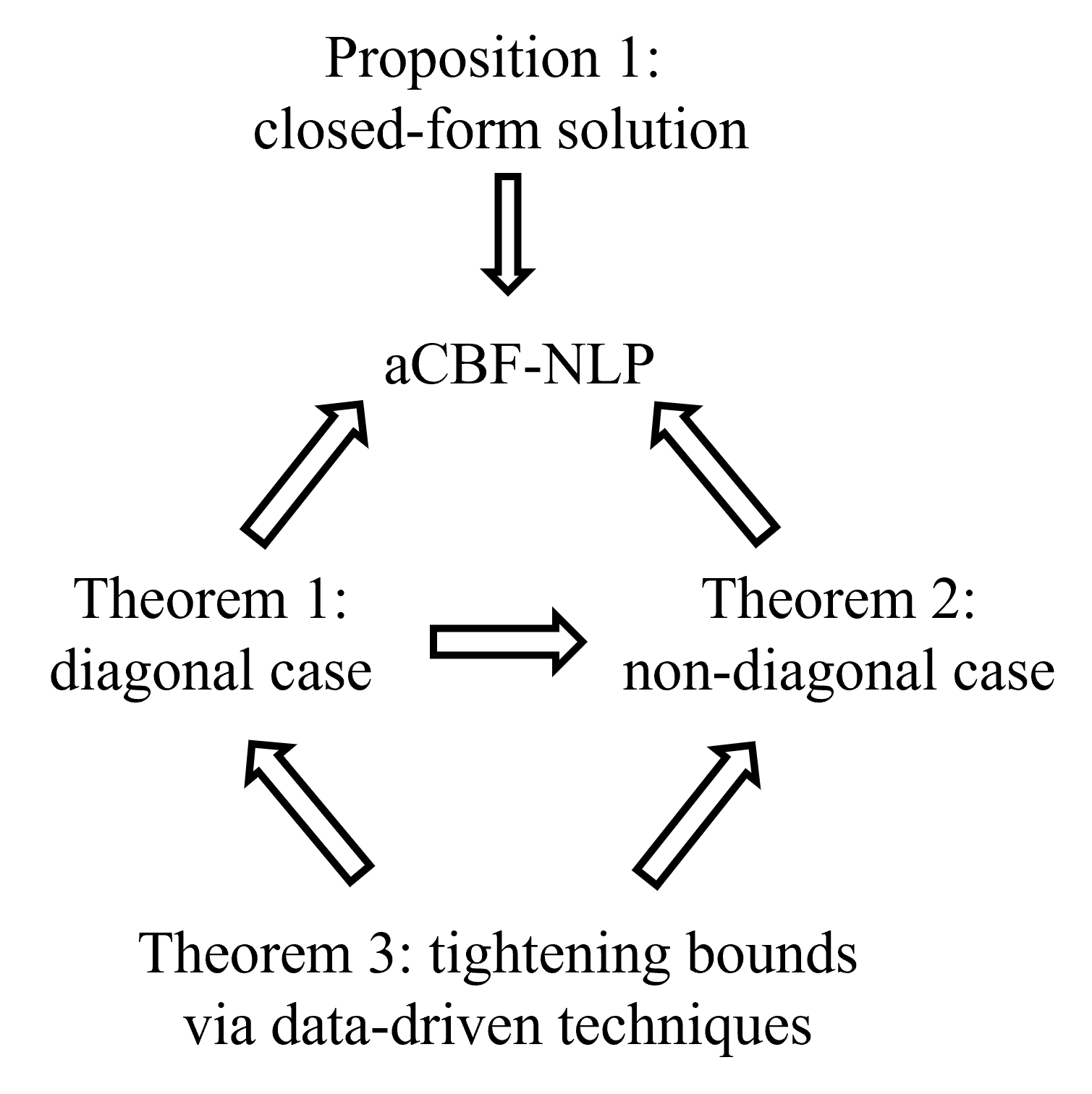}
\caption{Main results of this paper. }
\label{fig:illustraion}
\end{figure}

Inspired by the idea of adaptive control Lyapunov functions (aCLFs) \cite{KK95}, the adaptive CBF (aCBF) approach, which estimates the unknown parameters online to guarantee the safety of control affine systems with parametric uncertainties via a QP-based safe controller, is first proposed in \cite{taylor2020adaptive}.
In contrast to the aCLF-based stabilizing controller design, the aCBF-based safe control design is more challenging partially because the forward invariance of a predefined safe set must be ensured for all time and aCBFs do not  have the positive definiteness property possessed by aCLFs. Following the pioneering work of \cite{taylor2020adaptive}, various aCBF-based control methods are developed in the literature \cite{lopez2020robust,zhao2020adaptive,black2021fixed,isaly2021adaptive,cohen2022high,wang2022observer,wang2022robust,huang2022safety,azimi2021exponential,verginis2022funnel} and applied to several practical scenarios, such as adaptive cruise control \cite{taylor2020adaptive}, aircraft control \cite{lopez2020robust}, control of wing rock motion \cite{huang2022safety},  and control of unicycle vehicles \cite{verginis2022funnel}.
Most of these works only take into account parametric uncertainties in the drift term, while there are many physical systems that have parametric uncertainties in the control-input matrices, such as robotic systems with imprecise or time-varying mass and inertia parameters.
In \cite{azimi2021exponential}, a  filtering-based concurrent learning algorithm in the  CBF framework is proposed to design safe controllers for single-input-single-output systems with unknown control coefficients; the estimated parameter converges to the true value exponentially, but system safety is not guaranteed before the convergence of the parameter adaptations. In \cite{verginis2022funnel}, a zeroing CBF-based adaptive control algorithm is proposed to solve the funnel control problem for systems with parametrically uncertain control-input matrices, which can achieve tracking of a reference trajectory within a pre-defined funnel; however, this method may fail in singular configurations, as discussed in  Remark 1 of that paper. Despite these early contributions, the aCBF-based control design for systems with parametric uncertainties in control-input matrices is still an open field and merits further investigation.

Consider a control-affine system $\dot{x} = f(x) + g(x) u$ where $f$, $g$ include parametric uncertainties (e.g., $f$ and $g$ are identified by universal approximators such as neural networks). 
The main challenge of stabilizing such a system using adaptive controllers arises from the so-called ``loss of controllability'' problem; that is,  although the system is controllable, the identification model may lose its controllability at some points in time, owing to parameter adaptations \cite{bechlioulis2008robust,ioannou1996robust}.
The same issue could happen in the aCBF-based control design, which will result in the emptiness of the admissible safe control set and therefore, the infeasibility of the QP.
To the best of our knowledge, the \emph{singularity-free} aCBF-based safe controller is not yet developed in the literature, though
relevant stabilizing adaptive control schemes have been proposed in \cite{KKK95,xu2003robust,bechlioulis2008robust,ioannou1996robust}.
To bridge this gap, this paper proposes a singularity-free aCBF-based control design method for systems with parametric uncertainties in both $f$ and $g$. In contrast to the existing results (e.g., the approach developed in \cite{verginis2022funnel}) where the safety constraints (i.e., the CBF conditions) include  estimated parameters, the CBF condition of the proposed method only relies on the nominal values of the unknown parameters. Hence, the non-emptiness of the admissible safe control set can be verified in advance, and the singular configuration can be avoided. The safe control is obtained by solving a \emph{nonlinear program} (NLP), which has a closed-form solution. Furthermore, a data-driven approach is developed to reduce the potential conservatism of the proposed controller by tightening the parameter bounds.
The effectiveness of the proposed control strategy is demonstrated by numerical simulations. Main results of this paper are shown in Fig. \ref{fig:illustraion}.

The rest of this paper is structured as follows. In
Section \ref{sec:pre}, introduction to CBFs and the problem formulation are provided; in Section \ref{sec:main}, the proposed aCBF-based control approaches are presented; in Section \ref{sec:data}, a data-driven method that aims to reduce the conservatism of the proposed control methods is developed; in Section \ref{sec:simulation}, numerical simulation results that validate the proposed methods are presented; and finally, the conclusion is drawn in Section \ref{sec:conclusion}.

\section{Preliminaries \& Problem Statement}
\label{sec:pre}

\subsection{Notation}
For a positive integer $n$, denote $[n]=\{1,2,\cdots,n\}$. For a column vector $x\in\R^{n}$ or a row vector $x\in\R^{1\times n}$, $x_i$ denotes the $i$-th entry of $x$ and $\|x\|$ represents its 2-norm. For a given matrix $A\in\R^{n\times m}$, $A_{ij}$ denotes the $(i,j)$-th entry of the matrix $A$ and $\|A\|$ represents its Frobenius norm. Denote ${\bf 0}_{m}$ as a column vector of dimension $m$ whose entries are all zero, and ${\bf 0}_{m\times n}$ as a $m\times n$ matrix whose entries are all zero. Denote ${\rm diag}(a_1,a_2,\cdots,a_n)\in\R^{n\times n}$ as a diagonal matrix with diagonal entries $a_1,a_2,\cdots,a_n\in\R$.
Given vectors $x,y\in\R^n$, $x\leq y$ is satisfied in the entry-wise sense, i.e., $x_i\leq y_i$, $\forall i\in[n]$, and $x\odot y$ represents the Hadamard product (element-wise product) \cite{horn2012matrix}.
Denote the set of intervals on $\R$ by $\mathbb{IR}$, the set of $n$-dimensional interval
vectors by $\mathbb{IR}^n$, and the set of $n\times m$-dimensional interval matrices by $\mathbb{IR}^{n\times m}$.
The definition of interval operations, e.g., addition, substraction, multiplication, etc., follows those in \cite{moore2009introduction}.
Given two vectors $x,y\in\R^n$ and $x\leq y$, $[x,y]=\left[[x_1,y_1]\ \cdots\ [x_n,y_n]\right]^\top\in\mathbb{IR}^{n}$ represents an interval vector.
Consider the gradient $h_x\triangleq\frac{\pa h}{\pa x} \in \R^{n\times 1}$ as a row vector, where $x\in\R^n$ and $h:\R^n\to\R$ is a function with respect to $x$. 

\subsection{Control Barrier Function}
Consider a control affine system
\begin{equation}
\dot{x} = f(x) + g(x) u,\label{eqnsysp}
\end{equation}
where $x\in\mathbb{R}^n$ is the state, $u\in U\subset\mathbb{R}^m$ is the control input,  $f: \mathbb{R}^n\to\mathbb{R}^n$ and $g:\mathbb{R}^n\to\mathbb{R}^{n\times m}$ are locally Lipchitz continuous functions. 
Define a \emph{safe set} $\C=\{ x \in \R^n \mid h(x) \geq 0\}$
where $h$ is a continuously differentiable function.
The function $h$ is called a \emph{(zeroing) CBF} of relative degree 1, if  there exists a constant $\gamma>0$ such that $\sup_{u \in U}  \left[ L_f h(x) + L_g h(x) u + \gamma h(x)\right] \geq 0$
where $L_fh(x)=\frac{\pa h}{\pa x} f(x)$ and $L_gh(x)=\frac{\pa h}{\pa x} g(x)$ are Lie derivatives \cite{isidori1985nonlinear}. In this paper, we assume there is no constraint on the input $u$, i.e., $U=\R^m$. For any given $x\in\R^n$, the set of all control values that satisfy the CBF condition is defined as
$K(x) =  \{ u \in U \mid L_f h(x) + L_g h(x) u + \gamma h(x) \geq 0\}.$
It was proven in \cite{Xu2015ADHS} that any Lipschitz continuous controller $u(x) \in K(x)$ will guarantee the forward invariance of $\C$, i.e., the \emph{safety} of the closed-loop system. The provably safe control law is obtained by solving a convex QP that includes the CBF condition as its constraint.  The time-varying CBF with a general relative degree and its safety guarantee for a time-varying system are discussed in \cite{xu2018constrained}. 

\subsection{Problem Formulation}

Consider the following system:
\begin{IEEEeqnarray}{rCl}\label{eqnsys2}
\hspace{-2mm}
\begin{pmatrix}
    \dot x_1\\ \dot x_2
\end{pmatrix} &=
    f(x) \!+\!f_u(x) +
\begin{pmatrix}
    {\bf 0}_{m}\\ f_{\theta}(x)
\end{pmatrix}
+ \begin{pmatrix}
    {\bf 0}_{m\times n}\\ g(x) +g_{\lambda}(x)
\end{pmatrix}u,
\end{IEEEeqnarray}
where $x= \begin{pmatrix}x_1\\x_2\end{pmatrix}\in\R^{m+n}$ is the state with $x_1\in\R^m$ and  $x_2\in\R^n$,  $u\in\R^{n}$ is the control input,
$f:\R^{m+n}\to\R^{m+n}$ and $g:\R^{m+n}\to\R^{n\times n}$ are known Lipschitz functions, $f_u:\R^{m+n}\to\R^{m+n}$ is an unknown Lipschitz function, and $f_\ta:\R^{m+n}\to\R^{n}$ and $g_\la:\R^{m+n}\to\R^{n\times n}$ are parametric uncertainties.  We assume that $f_\ta$, $g$, and $g_\la$ have the following forms:
\begin{IEEEeqnarray}{rCl}
\IEEEyesnumber \label{eqnstructure}
\IEEEyessubnumber
f_\theta(x)&=&\begin{bmatrix}\theta_1^\top\varphi_1(x),&\theta_2^\top\varphi_2(x),&
 \cdots,&
 \theta_{n}^\top\varphi_{n}(x)
 \end{bmatrix}^\top,\\
 \IEEEyessubnumber
 g(x)&=&{\rm diag}(g_1(x),g_2(x),\cdots,g_n(x)),\\
 \IEEEyessubnumber
 g_\la(x)&=&{\rm diag}(\lambda_1^\top\psi_1(x),\lambda_2^\top\psi_2(x),\cdots,\la_{n}^\top\psi_{n}(x)),
\end{IEEEeqnarray}
where $g_i:\R^{m+n}\to\R$ is a known Lipschitz function, $\ta_i\in\R^{p_i}$ and $\la_{i}\in\R^{q_i}$ are unknown parameters, and $\varphi_i:\R^{m+n}\to\R^{p_i}$ and $\psi_i:\R^{m+n}\to\R^{q_i}$ are known Lipschitz functions (regressors) with $p_i,q_i$ appropriate positive integers and $i\in[n]$. Note that the functions $f,g,f_u,\varphi_i, \psi_i$, $i\in[n]$, are assumed to be Lipschitz continuous to ensure the existence and uniqueness of the solution to \eqref{eqnsys2}.
Define a safe set $\C\subset\R^{m+n}$ as
\begin{equation}\label{setc}
    \C=\{x: h(x)\geq 0\},
\end{equation}
where $h:\R^{m+n}\to\R$ is a continuously differentiable function.
We also make the following two assumptions on the boundedness of the unknown function $f_u$ and the unknown parameters $\ta_i,\la_i$.

\begin{assumption}\label{assump:0}
There exist known functions $\underline{f}_u(x)$, $\overline{f}_u(x):\R^{m+n}\to\R^{m+n}$ such that $\underline{f}_u(x)\leq f_u(x)\leq\overline{f}_u(x)$.
\end{assumption}

\begin{assumption}\label{assump:1}
For every $i\in[n]$, there exist known vectors $\overline{\theta}_i, \underline{\theta}_i\in\mathbb{R}^{p_i}$ and $\overline{\la}_i,\underline{\la}_i\in\mathbb{R}^{q_i}$ such that  $\underline{\theta}_i\leq\theta_i \leq\overline{\theta}_i$ and  $\underline{\lambda}_i\leq \la_i\leq\overline{\lambda}_i$.
\end{assumption}

\begin{remark}\label{remark1}
In the adaptive stabilizing control design problem,  bounds for the unknown parameters as given in Assumption \ref{assump:1} are not necessarily required to be known since the asymptotic stability of the closed-loop system can be proven using Barbalat's lemma  when the derivative of the Lyapunov function is negative semi-definite \cite{khalil2002nonlinear}. Because CBFs do not have the favourable positive definiteness property as Lyapunov functions, the CBF-based safe control design is more challenging. Although an aCBF-based control approach is proposed in \cite{taylor2020adaptive} without assuming boundedness of the unknown parameters, its performance is conservative as the system only operates in a subset of the original safety set. In \cite{lopez2020robust}, a robust aCBF-based controller is developed under the assumption that is similar to Assumption \ref{assump:1}, i.e.,
the unknown parameters and the parameter estimation error both belong to known closed convex sets; however, the system model considered there does not include the parametric uncertainty $g_\la$ in the control-input matrix.
\end{remark}

The main problem that will be investigated in this paper is stated as follows.
\begin{problem}\label{prob1}
Consider the system \eqref{eqnsys2} with $f_\ta$, $g$, and $g_\la$ given in \eqref{eqnstructure} and the safe set defined in \eqref{setc} where $h$ has a relative degree 1. Suppose that Assumptions \ref{assump:0} and \ref{assump:1} hold. Design a feedback controller $u$ such that the closed-loop system is always safe, i.e., $h(x(t))\geq 0$ for all $t\geq 0$.
\end{problem}

We will propose an aCBF-NLP-based method for solving Problem \ref{prob1} in Section \ref{sec:mainadaptive} and generalize it to the case where $g$ and $g_\la$ are non-diagonal in Section \ref{sec:maingeneral}.
Moreover, although we only consider the CBF $h$ with a relative degree 1 in this work, our results can be easily extended to the higher relative degree cases by using techniques in \cite{xu2018constrained,nguyen2016exponential,tan2021high}; a mass-spring system that has a relative degree 2 will be shown in Example \ref{example3} of Section \ref{sec:simulation}.

\section{aCBF-NLP-based Safe Control Design}
\label{sec:main}
In this section, the main result of this work will be presented. In Section \ref{sec:mainadaptive}, an aCBF-NLP-based safe control design approach will be proposed to solve Problem \ref{prob1}; in Section \ref{sec:nlpsolution}, the closed-form solution to the NLP will be presented; in Section \ref{sec:maingeneral}, the proposed method is extended to a more general class of systems.

\subsection{aCBF-NLP-Based Control Design}
\label{sec:mainadaptive}

In this subsection, an aCBF-NLP-based control design method is proposed to solve Problem \ref{prob1}. Recall that $f_\ta$, $g$, $g_\la$ have the forms given in \eqref{eqnstructure} where $\theta_i\in\R^{p_i}$ and $\la_i\in\R^{q_i}$ are unknown parameters.
We choose arbitrary values $\theta^0_i\in\R^{p_i}$ and $\la^0_i\in\R^{q_i}$  satisfying $ \underline{\theta}_i\leq\theta^0_i\leq\overline{\theta}_i$ and $ \underline{\la}_i\leq\la^0_i\leq\overline{\la}_i$ as the nominal values for $\theta_i$ and $\la_i$, respectively. Furthermore, we define
\begin{align}\label{eqmunv}
\mu_i\triangleq\|\theta_i-\theta^0_i\|, \quad \nu_i\triangleq\|\la_i-\la_i^0\|,\quad  \forall i\in[n].
\end{align}
According to Assumption \ref{assump:1} and the definition of 2-norm, 
\begin{align*}
\mu_i\!\leq\! \bar{\mu}_i\!\triangleq\! \sqrt{\!\sum_{j=1}^{p_i}\max\{((\overline{\theta}_i)_j\!-\!({\theta}_i^0)_j)^2,((\underline{\theta}_i)_j\!-\!({\theta}_i^0)_j)^2\}},\\
\nu_i\!\leq\! \bar{\nu}_i\!\triangleq\!\sqrt{\!\sum_{j=1}^{q_i}\max\{((\overline{\la}_i)_j\!-\!({\la}_i^0)_j)^2,((\underline{\la}_i)_j\!-\!({\la}_i^0)_j)^2\}},
\end{align*}
where $(\overline \theta_i)_j$, $(\underline  \theta_i)_j$, $(\overline \la_i)_j$, $(\underline  \la_i)_j$ denote the $j$-th entry of $\overline\theta_i$, $\underline\theta_i$, $\overline \la_i$, $\underline\la_i$, respectively. Note that in this paper the adaptive laws are used to estimate parameters $\mu_i$ and $\nu_i$, which are \emph{scalars}, rather than parameters $\theta_i$ and $\la_i$, which are vectors. The following assumption assumes that each
diagonal entry of $g(x)+g_\la(x)$ is away from zero.

\begin{assumption}\label{assump:2}
Given functions $g(x),g_\la(x)$ in  diagonal forms as shown in \eqref{eqnstructure}, there exist constants $b_1,...,b_n>0$ such that $g_i(x)+\la_i^\top\psi_i(x)$ 
satisfies $|g_i(x)+\la_i^\top\psi_i(x)|\geq b_i$ for any $i\in[n]$ and any $x\in\C$. Moreover, the sign of $g_i(x)+\la_i^\top\psi_i(x)$ is known, and without loss of generality, it is assumed that $g_i(x)+\la_i^\top\psi_i(x)>0$ for any $i\in[n]$ and $x\in\C$.
\end{assumption}

\begin{remark}
The condition $|g_i(x)+\la_i^\top\psi_i(x)|\geq b_i$ in Assumption \ref{assump:2} is imposed to avoid the loss of controllability problem \cite{xu2003robust,bechlioulis2008robust}. In Section \ref{sec:maingeneral}, Assumption \ref{assump:2} is relaxed to Assumption \ref{assump:4} for a more general class of systems (i.e., $g(x),g_\la(x)$ are not diagonal).  However, the safe controller constructed under Assumption \ref{assump:2} (cf. Theorem \ref{theorem1}) tends to have a less conservative performance than that under Assumption \ref{assump:4} (cf. Theorem \ref{theorem3}); see Remark \ref{remark:generaldisadvantage} and Example \ref{example3} for more details.
\end{remark}

The following theorem shows an aCBF-based controller that ensures the safety of system \eqref{eqnsys2}.

\begin{theorem}\label{theorem1}
Consider the system \eqref{eqnsys2} with $f_\theta,g,g_\la$ specified in \eqref{eqnstructure} and the safe set $\C$ defined in \eqref{setc}. Suppose that\\
(i) Assumptions \ref{assump:0}, \ref{assump:1} and \ref{assump:2} hold;\\
(ii) There exist positive constants $\gamma,\epsilon_1,\epsilon_2,\gamma^\theta_i,\gamma^{\lambda}_i>0$ where $i\in[n]$, such that the following set is non-empty:
\begin{equation}\label{cbfcondition1}
K_{BF}(x)\triangleq\left\{{\mathfrak u}\in\R^n\mid \Psi_{0}(x)+\Psi_{1}(x) {\mathfrak u}\geq 0\right\},\;\forall x\in \C,
\end{equation}
where $\Psi_{0}(x)=\mathscr{M}+\sum_{i=1}^{n}h_{x_2,i}\theta_i^{0\top}\varphi_i-n(\ep_1+\ep_2)+\ga\left[h-\sum_{i=1}^{n}\left(\frac{\bar{\mu}_{i}^2}{2\ga_{i}^\theta}+\frac{\bar{\nu}_{i}^2}{2\ga_{i}^\la}\right)\right]$,
$\Psi_{1}(x)=[h_{x_2,1}^2(g_{1}+\la_1^{0\top} \psi_1) \\ h_{x_2,2}^2(g_{2}+\la_2^{0\top} \psi_2)\ \cdots \ h_{x_2,n}^2(g_n+\la_{n}^{0\top} \psi_{n})]$,
$\mathscr{M}=h_{x}f+\sum_{j=1}^{m+n}\min\{h_{x,j}\underline{f}_{u,j},h_{x,j}\overline{f}_{u,j}\}$, $h_x=\frac{\pa  h}{\pa x}$, $h_{x_2}=\frac{\pa h}{\pa x_2}$, and $h_{x,i}$, $h_{x_2,i}$, $\underline{f}_{u,i}$, $\overline{f}_{u,i}$ denotes the $i$-th entry of $h_x$, $h_{x_2}$, $\underline{f}_u$, $\overline{f}_u$, $i\in[n]$, respectively;\\
(iii) For any $i\in[n]$, $\hat\mu_i$ and $\hat\nu_i$ are estimated parameters governed by the following adaptive laws:
\begin{subequations}\label{adaptive:both}
\begin{align}
\dot {\hat{\mu}}_{i} &=-{\gamma} \hat\mu_{i}+\gamma^{\theta}_i|h_{x_2,i}|\|\varphi_i\|, \label{adaptivelaw1}\\
\dot {\hat{\nu}}_{i} &=-{\gamma} \hat\nu_{i}+\gamma^{\lambda}_ih_{x_2,i}^2|u_{0,i}|\|\psi_i\|,\label{adaptivelaw2}
\end{align}
\end{subequations}
where $\hat\mu_{i}(0)>0,\hat\nu_{i}(0)>0$ and $u_0=[u_{0,1},\cdots,u_{0,n}]^\top$ is a Lipschitz function satisfying $u_0 \in K_{BF}(x)$;\\
(iv) The following inequality holds:

$
h(x(0))\geq \sum_{i=1}^{n}\left( \frac{\hat\mu_i(0)^2+\bar{\mu}_{i}^2}{2\gamma^{\theta}_i}+\frac{\hat\nu_i(0)^2+\bar{\nu}_{i}^2}{2\gamma^{\lambda}_i} \right)
$;\\
\noindent Then, the control input $u=h_{x_2}^\top\odot s(u_0)\in\R^n$ will make $h(x(t))\geq 0$ for $t>0$, where
$s(u_0)\triangleq[s_1(u_{0,1}),s_2(u_{0,2}),\cdots,s_n(u_{0,n})]^\top$ and
\begin{equation}
s_i(u_{0,i})\triangleq u_{0,i}+\frac{\kappa_{1,i}}{b_i}+\frac{\kappa_{2,i}^2u^2_{0,i}}{b_i(\kappa_{2,i}|h_{x_2,i}||u_{0,i}|+\epsilon_2)},\label{adapcontrol}
\end{equation}
with $\kappa_{1,i}=\frac{\hat\mu_{i}^2\|\varphi_i\|^2}{\hat\mu_{i}\|\varphi_i\||h_{x_2,i}|+\epsilon_1}$, $\kappa_{2,i}=\hat\nu_{i}\|\psi_i\||h_{x_2,i}|$, $i\in[n]$.
\end{theorem}
\begin{proof}
From \eqref{adaptive:both}, $\dot {\hat{\mu}}_{i} \geq -{\gamma} \hat\mu_{i},\dot {\hat{\nu}}_{i} \geq -{\gamma} \hat\nu_{i}$ hold. Since $\hat\mu_{i}(0)>0,\hat\nu_{i}(0)>0$, it is easy to see that $\hat\mu_{i}(t)\geq 0$ and $\hat\nu_{i}(t)\geq 0$ for any $t>0$ by the Comparison Lemma \cite[Lemma 2.5]{khalil2002nonlinear}. Define a new candidate CBF $\bar h$ as
$\bar{h}(x,t)=h(x)-\sum_{i=1}^{n}\bigg(\frac{\tilde{\mu}_{i}^2}{2\gamma^{\theta}_i} +\frac{\tilde{\nu}_{i}^2}{2\gamma^{\lambda}_i} \bigg)$,
where $\tilde{\mu}_{i}={\mu}_{i}-\hat{\mu}_{i}$ and $\tilde{\nu}_i={\nu}_i-\hat{\nu}_i$.
It can be seen that $\bar{h}(x(0),0) = h(x(0))-\sum_{i=1}^{n}\bigg( \frac{(\mu_i-\hat\mu_i(0))^2}{2\gamma^{\theta}_i} +\frac{(\nu_i-\hat\nu_i(0))^2}{2\gamma^{\lambda}_i} \bigg)
\geq h(x(0))-\sum_{i=1}^{n}\bigg( \frac{\mu_i^2+\hat\mu_i^2(0)}{2\gamma^{\theta}_i} +\frac{\nu_i^2+\hat\nu_i^2(0)}{2\gamma^{\lambda}_i} \bigg)
\geq h(x(0))-\sum_{i=1}^{n}\bigg( \frac{\bar\mu_i^2+\hat\mu_i^2(0)}{2\gamma^{\theta}_i} +\frac{\bar \nu_i^2+\hat\nu_i^2(0)}{2\gamma^{\lambda}_i} \bigg)\geq 0$, where the first inequality comes from the fact that $\mu_i,\nu_i,\hat\mu_i(0),\hat\nu_i(0)\geq 0$, the second one arises from the definitions of $\bar{\mu}_i,\bar{\nu}_i$, 
and the last one holds because of condition (iv).

We claim that $\dot{\bar{h}}\geq-\gamma \bar h$ where $\dot{\bar{h}}$ is the time derivative of $\bar h$. Indeed, $\dot{\bar{h}}
=h_x(f+f_u)+ \sum_{i=1}^{n} (h_{x_2,i}(\theta_i^\top\varphi_i+(g_{i}+\la_i^\top \psi_i)u_i)+\frac{\tilde{\mu}_{i}\dot{\hat\mu}_{i}}{\gamma^{\theta}_i} +\frac{\tilde{\nu}_{i}\dot{\hat{\nu}}_{i}}{\gamma^{\lambda}_i} )
\geq \mathscr{M}+\sum_{i=1}^{n} \left(h_{x_2,i}(\theta_i^\top\varphi_i+(g_{i}+\la_i^\top \psi_i)u_i)+\frac{\tilde{\mu}_{i}\dot{\hat{\mu}}_{i}}{\gamma^{\theta}_i} +\frac{\tilde{\nu}_{i}\dot{\hat{\nu}}_{i}}{\gamma^{\lambda}_i} \right)$. Substituting \eqref{adapcontrol} into the inequality above and recalling Assumption \ref{assump:2}, we have
\begin{IEEEeqnarray}{rCl}
\dot{\bar{h}}
&\geq&\mathscr{M}+\sum_{i=1}^{n}\bigg(h_{x_2,i}\theta_i^\top\varphi_i+\frac{\tilde{\mu}_{i}\dot{\hat{\mu}}_{i}}{\gamma^{\theta}_i} +\frac{\tilde{\nu}_{i}\dot{\hat{\nu}}_{i}}{\gamma^{\lambda}_i}\bigg)\nonumber\\
&&\!+\!\sum_{i=1}^{n}\!\!h_{x_2,i}^2\bigg((g_i\!+\!\la_i^\top \psi_i)u_{0,i}\!+\!\kappa_{1,i}\!+\!\frac{\kappa_{2,i}^2u^2_{0,i}}{\kappa_{2,i}|h_{x_2,i}||u_{0,i}|\!+\!\epsilon_2}\bigg)\nonumber\\
&\geq&\mathscr{M}+\sum_{i=1}^{n}\left(h_{x_2,i}\theta_i^{0\top}\varphi_i+h_{x_2,i}^2(g_{i}+\la_i^{0\top} \psi_i)u_{0,i}\right)\nonumber\\
&&\!+\!\sum_{i=1}^{n}\!\left(\!\frac{\tilde{\mu}_{i}\dot{\hat{\mu}}_{i}}{\gamma^{\theta}_i} \!+\!\frac{\tilde{\nu}_{i}\dot{\hat{\nu}}_{i}}{\gamma^{\lambda}_i}\! +\!h_{x_2,i}^2\!\left(\!\kappa_{1,i}\!+\!\frac{\kappa_{2,i}^2u^2_{0,i}}{\kappa_{2,i}|h_{x_2,i}||u_{0,i}|\!+\!\epsilon_2}\!\right)\!\right)\nonumber\\
&&+\sum_{i=1}^{n}(h_{x_2,i}
(\theta_i-\theta_i^0)^\top\varphi_i+h^2_{x_2,i}(\la_i-\la_i^0)^\top\psi_i u_{0,i}
)\nonumber\\
&\geq& \Psi_0+\Psi_1 u_0+n(\ep_1+\ep_2)-\ga\bigg[h-\sum_{i=1}^{n}\bigg(\frac{\bar{\mu}_{i}^2}{2\ga_{i}^\theta}+\frac{\bar{\nu}_{i}^2}{2\ga_{i}^\la}\bigg)\bigg]\nonumber\\
&&+\sum_{i=1}^{n}\bigg(\!\frac{\tilde{\mu}_{i}\dot{\hat{\mu}}_{i}}{\gamma^{\theta}_i} \!+\!\frac{\tilde{\nu}_{i}\dot{\hat{\nu}}_i}{\gamma^{\lambda}_i}\! \bigg)+\sum_{i=1}^{n}\bigg(\!-\mu_{i}\|\varphi_i\||h_{x_2,i}|\!+\!\kappa_{1,i}h_{x_2,i}^2\nonumber\\
&&-\nu_{i}\|\psi_i\||u_{0,i}|h_{x_2,i}^2+\frac{\kappa_{2,i}^2 u_{0,i}^2 h_{x_2,i}^2}{\kappa_{2,i}|h_{x_2,i}||u_{0,i}|+\epsilon_2}\bigg),\label{dotbarh11}
\end{IEEEeqnarray}
where the third inequality arises from Cauchy–Schwarz inequality. It is easy to check that $-\hat\mu_{i}\|\varphi_i\||h_{x_2,i}|+\kappa_{1,i}h_{x_2,i}^2=-\frac{\hat\mu_{i}\|\varphi_i\||h_{x_2,i}|\epsilon_1}{\hat\mu_{i}\|\varphi_i\||h_{x_2,i}|+\epsilon_1}\geq -\epsilon_1$ and $-\hat\nu_{i}\|\psi_i\||u_{0,i}|h_{x_2,i}^2+\frac{\kappa_{2,i}^2 u_{0,i}^2 h_{x_2,i}^2}{\kappa_{2,i}|h_{x_2,i}||u_{0,i}|+\epsilon_2}=-\frac{\kappa_{2,i}|h_{x_2,i}||u_{0,i}|\epsilon_2}{\kappa_{2,i}|h_{x_2,i}||u_{0,i}|+\epsilon_2}\geq -\epsilon_2$; furthermore, $\Psi_0+\Psi_1 u_0\geq 0$ because $u_0\in K_{BF}$. Based on these two facts and recalling that $\mu_i=\tilde\mu_i+\hat\mu$, ${\nu}_i=\tilde{\nu}_i+\hat{\nu}_i$, one can see that \eqref{dotbarh11} can be expressed as $\dot{\bar{h}}\geq\sum_{i=1}^{n}\!\left(\!
\tilde{\mu}_{i}\!\bigg(\!\frac{1}{\gamma_{i}^\theta}\dot{\hat{\mu}}_{i}\!-\!\|\varphi_i\||h_{x_2,i}|\!\right)\!+\!\tilde{\nu}_{i}\left(\!\frac{1}{\gamma_{i}^\la}\dot{\hat{\nu}}_{i}\!-\!\|\psi_i\||u_{0,i}|h_{x_2,i}^2\!\bigg)\!
\right)-\ga\left[h-\sum_{i=1}^{n}\left(\frac{\bar{\mu}_{i}^2}{2\ga_{i}^\theta}+\frac{\bar{\nu}_{i}^2}{2\ga_{i}^\la}\right)\right]$. Substituting \eqref{adaptive:both} into the inequality above yields $\dot{\bar{h}}\geq -\ga\sum_{i=1}^{n}\left(\frac{\tilde{\mu}_{i}\hat{\mu}_{i}}{\ga_{i}^\theta} +\frac{\tilde{\nu}_{i}\hat{\nu}_{i}}{\ga_{i}^\la} \right)-\ga\left[h-\sum_{i=1}^{n}\left(\frac{\bar{\mu}_{i}^2}{2\ga_{i}^\theta}+\frac{\bar{\nu}_{i}^2}{2\ga_{i}^\la}\right)\right]$. Since $\hat{\mu}_{i}\tilde{\mu}_{i}=({\mu}_{i}-\tilde{\mu}_{i})\tilde{\mu}_{i}\leq \frac{\mu_{i}^2-\tilde{\mu}_{i}^2}{2}\leq \frac{\bar{\mu}_{i}^2-\tilde{\mu}_{i}^2}{2}$ and $\hat{\nu}_{i}\tilde{\nu}_{i}=({\nu}_{i}-\tilde{\nu}_{i})\tilde{\nu}_{i}\leq \frac{{\nu}_{i}^2-\tilde{\nu}_{i}^2}{2}\leq \frac{\bar{\nu}_{i}^2-\tilde{\nu}_{i}^2}{2}$,
we have $\dot{\bar{h}}\geq-\gamma\left[h-\sum_{i=1}^{n}\left(\frac{\tilde{\mu}_{i}^2}{2\gamma^{\theta}_i} +\frac{\tilde{\nu}_{i}^2}{2\gamma^{\lambda}_i} \right)\right]=-\gamma \bar h$, which shows the correctness of the claim.

Because $\bar h(x(0),0)\geq 0$, it is easy to see that $\bar h(t)\geq 0$ for $t>0$.  Since $\bar h\leq h$ by definition, we have $h(t)\geq 0$ for $t>0$, which completes the proof.
\end{proof}

\begin{remark}\label{remark2Kbf}
It should be noticed that the CBF condition $\Psi_{0}(x)+\Psi_{1}(x) {\mathfrak u}\geq 0$  shown in \eqref{cbfcondition1} is imposed on the intermediate variable, $u_0$,  instead of the real control input, $u$. Furthermore, the CBF condition \eqref{cbfcondition1} only relies on the nominal values of the unknown parameters, which implies that the CBF condition (i.e., the non-emptiness of the set  $K_{BF}$) can be verified conveniently by selecting the variables in Condition (ii) appropriately. 

In \cite{KK95}, the problem of adaptive stabilization of a nonlinear system is converted to the
nonadaptive stabilization of a modified system by utilizing an aCLF. While the idea of \cite{KK95} 
may be extended to develop an aCBF-based safe control law for \eqref{eqnsys2}, 
the resulting CBF condition would need to be verified for any $\theta_i$ and $\la_i$ satisfying $\underline{\theta}_i\leq\theta_i\leq\overline{\theta}_i$ and $\underline{\la}_i\leq\la\leq\overline{\la}_i$, which is much more restrictive than the CBF condition given in Theorem \ref{theorem1} above. On the other hand, 
the CBF condition given in \cite{verginis2022funnel} relies on  estimated parameters (i.e., $\Psi_0$ and $\Psi_1$ are functions of the estimated parameters in
the adaptive laws), which renders the singular configuration (i.e., the set  $K_{BF}$ is empty) difficult to verify; see the discussion in Remark 1 of \cite{verginis2022funnel}.
\end{remark}

\begin{remark}
The number of ODEs for parameter estimation in Theorem \ref{theorem1} is much less than that in other aCBF-based approaches such as \cite{taylor2020adaptive,lopez2020robust,verginis2022funnel}. As can be seen from the adaptive laws shown in \eqref{adaptive:both}, our method only requires solving $2n$ ODEs that estimate \emph{scalars} $\mu_i\in\R$ and $\nu_i\in\R$ for $i\in[n]$ (cf.  \eqref{eqmunv}); in contrast, other aCBF methods have to estimate the original unknown parameters $\theta_i\in\R^{p_i}$ and $\la_i\in\R^{q_i}$ for $i\in[n]$, which results in a total of $2n\sum_{j=1}^{n}(p_j+q_j)$ ODEs. This reduction of number of ODEs is particularly useful when $p_i$ and $q_i$ are large, e.g., when $\theta_i$ and $\la_i$ are weights of deep neural networks.
\end{remark}

The safe control law $u(x)\triangleq[u_1,\dots,u_n]^\top$ in Theorem \ref{theorem1} can be obtained pointwise for any $x\in\C$. Specifically, each $u_i,i\in[n],$ can be obtained by solving the following  optimization problem:
\begin{align}
\min_{u_{i}\in\R} \quad & (u_i-u_{d,i})^2\label{cbfnlp1}\tag{aCBF-NLP}\\
\textrm{s.t.}\quad &  \Phi^i_0(x)+\Phi^i_1(x) u_{0,i}\geq 0,\nonumber\\
\quad & u_i=h_{x_2,i} s_i(u_{0,i}),\nonumber
\end{align}
where $s_i(\cdot)$ is the function defined in \eqref{adapcontrol}, $u_{d,i}$ is the $i$-th entry of the nominal controller,
\begin{IEEEeqnarray}{rCl}
\IEEEyesnumber \label{psij}
\IEEEyessubnumber \label{psi0j}
\Phi_0^i &=&\begin{cases} \frac{\rho_ih_{x_2,i}^2|g_i+\la_i^{0\top} \psi_i|}{\sum_{j=1}^{n}\rho_jh_{x_2,j}^2|g_j+\la_j^{0\top} \psi_j|}\Psi_0, & \ {\rm if} \ \Psi_1\neq 0,\\
\Psi_0/n, & \ {\rm otherwise},
\end{cases}\\
\IEEEyessubnumber \label{psi1j}
\Phi_1^i &=& h_{x_2,i}^2(g_i+\la_i^{0\top} \psi_i),
\end{IEEEeqnarray}
with $\Psi_0$, $\Psi_1$ defined in Theorem \ref{theorem1}, and $\rho_i>0, i\in[n],$ are tuning parameters. Note that \eqref{psi0j} is well-defined as $\sum_{j=1}^{n}\rho_jh_{x,j}^2|g_{j}+\la_j^{0\top} \psi_j|\neq 0$ if $\Psi_1\neq 0$ and $\sum_{i=1}^n\Phi_0^i=\Psi_0$, $\sum_{i=1}^n\Phi_1^iu_{0,i}=\Psi_1 u_0$.

Different from the traditional CBF-QP formulation \cite{ames2016control,Xu2015ADHS}, the optimization \eqref{cbfnlp1} is an NLP because of the nonlinear function $s_i(\cdot)$. Solving an NLP is computationally challenging in general; however, optimization \eqref{cbfnlp1} has a \emph{closed-form} solution, which will be discussed in the next subsection.

\begin{remark}
An alternative optimization to obtain the safe control law $u(x)$ can be formulated as:
\begin{align}
\min_{u\in\R^n} \quad & \|u-u_d\|^2\label{cbfnlp0}\\
\textrm{s.t.}\quad &  \Psi_0+\Psi_1 u_0\geq 0,\nonumber\\
\quad & u=h_{x_2}^\top\odot s(u_0),\nonumber
\end{align}
where $\Psi_0,\Psi_1,s(\cdot)$ are defined in Theorem \ref{theorem1} and $u_d$ is the nominal controller. The admissible set of $u_0$ in \eqref{cbfnlp0} is larger than that of \eqref{cbfnlp1}, but the existence of a closed-form solution to \eqref{cbfnlp0} is still unclear to us.
\end{remark}

\begin{remark}
The main idea behind the formulation of \eqref{cbfnlp1} is to split the set $K_{BF}$ into $n$ independent set $K_{BF}^i\triangleq \{{\mathfrak u}\in\R: \Phi_0^i+\Phi_1^i{\mathfrak u} \geq 0\}$,
such that $u_{i}\in K_{BF}^i,\forall i\in[n]\implies u\in K_{BF}$.
It is easy to see that if $K_{BF}\neq \emptyset$, then $K_{BF}^i\neq\emptyset$ for any $i\in[n]$ and any $x\in\C$: if $\Psi_1\neq 0$, then $\Phi_1^i=0\implies\Phi_0^i=0\implies\Phi^i_0(x)+\Phi^i_1(x) u_{0,i}\geq 0$ always holds; if $\Psi_1=0$ and $K_{BF}\neq \emptyset\implies$ $\Phi^i_1(x)=0$ and $\Psi_0\geq 0\implies\Phi_0^i=\frac{\Psi_0}{n}\geq 0\implies\Phi^i_0(x)+\Phi^i_1(x) u_{0,i}\geq 0$ always holds.
\end{remark}

\subsection{Closed-form Solution to the aCBF-NLP}
\label{sec:nlpsolution}
In this subsection, we will discuss the closed-form solution to \eqref{cbfnlp1}. We will focus on the case $n=1$ because the $n>1$ case can be easily solved by considering the $n$ NLPs in \eqref{cbfnlp1} independently.

When  $n=1$, the subscript $i$ for all relevant variables defined in Theorem \ref{theorem1} will be discarded for the sake of simplicity.
It is also easy to see that $\Phi_0=\Psi_1$, $\Phi_1=\Psi_1$, and $h_{x_2}=0\implies u=0$ according to Theorem \ref{theorem1}. Thus, without loss of generality, we assume that $h_{x_2}\neq 0$ in the analysis of this subsection. By substituting $u=h_{x_2}s(u_0)$ into the objective function of \eqref{cbfnlp1}, it is easy to see that \eqref{cbfnlp1} is equivalent to the following optimization when $n=1$:
\begin{align}
\min_{u_0\in\R} \quad & (s(u_0)- \bar u_d)^2\label{cbfnlp}\\
\textrm{s.t.}\quad &  \Psi_0+\Psi_1 u_0\geq 0,\nonumber
\end{align}
where $\bar u_d=u_d/h_{x_2}$ and $\Psi_0,\Psi_1,s(\cdot)$ are defined in Theorem \ref{theorem1}.
Based on the properties of the function $s(\cdot)$ presented in Lemma \ref{lemma:1} (see Appendix), the optimal solution to \eqref{cbfnlp} can be obtained, from which the closed-form solution to \eqref{cbfnlp1} can be obtained, as shown in the following proposition.
\begin{proposition}\label{theorem2}
The closed-form solution to \eqref{cbfnlp1} can be represented as
\begin{equation*}
    u=\begin{cases}
     h_{x_2} \max \left( s\left(-\frac{\Psi_0}{\Psi_1}\right), \bar u_d\right),  & {\rm if} \ \A_1\ {\rm holds},\\
     h_{x_2}\min \left(s\left(-\frac{\Psi_0}{\Psi_1}\right),\bar u_d\right),  & {\rm if} \ \A_2\ {\rm holds},\\
      h_{x_2}\max(s(y^*), \bar u_d), & {\rm if}\  \A_3\ {\rm holds},\\
      u_d, & {\rm if}\ \Psi_1=0\wedge\bar b-\ka_2\geq 0,\\
     0, &{\rm if} \ h_{x_2}=0,
    \end{cases}
\end{equation*}
where $\kappa_2,\Psi_0,\Psi_1$ are given in Theorem \ref{theorem1}, $\bar b=b|h_{x_2}|$, $\bar\ep_2=\ep_2/|h_{x_2}|$, $y^*=\frac{\bar\ep_2[(\ka_2-\bar b)-\sqrt{\ka_2(\ka_2-\bar b)}]}{\ka_2(\ka_2-\bar b)}$, $\A_1: h_{x_2} \neq 0\wedge ((\bar b -\ka_2\geq 0\wedge \Psi_1>0)\vee (\bar b-\ka_2<0\wedge \Psi_0+\Psi_1 y^*< 0))$, $\A_2: h_{x_2}\neq 0\wedge \bar b-\ka_2\geq 0\wedge \Psi_1<0$, and $\A_3: h_{x_2}\neq 0\wedge \bar b-\ka_2< 0\wedge  \Psi_0+\Psi_1 y^*\geq 0$.

\end{proposition}
\begin{proof}
Note that similar to the aCBF-QPs presented in \cite{taylor2020adaptive,lopez2020robust}, the optimization \eqref{cbfnlp1} is solved pointwise for a given $(x,\hat\mu,\hat\nu)$, such that $\ka_1$, $\ka_2$ defined in \eqref{adapcontrol} and $h_{x_2}$  should be considered as constants when solving \eqref{cbfnlp}.
If $\bar b-\ka_2\geq 0$, $s(y)$ is monotonically increasing, according to Lemma \ref{lemma:1}. When $\Psi_1>0$, one can see that $K_{BF}=\{{\mathfrak u}:{\mathfrak u}\geq -\frac{\Psi_0}{\Psi_1}\}$ and $s(u_0)\in\left[s\left(-\frac{\Psi_0}{\Psi_1}\right),+\infty\right]$ for any $u_0\in K_{BF}$. It is easy to verify that $s(u_0^*)=\bar u_d$ if $\bar u_d\geq s\left(-\frac{\Psi_0}{\Psi_1}\right)$ and $s(u_0^*)= s\left(-\frac{\Psi_0}{\Psi_1}\right)$ when $\bar u_d<s\left(-\frac{\Psi_0}{\Psi_1}\right)$, where $u_0^*$ denotes the solution to \eqref{cbfnlp}.
Hence, in conclusion, one has $s(u_0^*)=\max \left\{s\left(-\frac{\Psi_0}{\Psi_1}\right),\bar u_d\right\}$, such that the closed-form solution to \eqref{cbfnlp1} is $u=h_{x_2}\max \left\{s\left(-\frac{\Psi_0}{\Psi_1}\right),\bar u_d\right\}$. Performing the similar analysis one can see that the closed-form solution to \eqref{cbfnlp1} is $u=h_{x_2}\min \left\{s\left(-\frac{\Psi_0}{\Psi_1}\right),\bar u_d\right\}$ when $\Psi_1<0$. If $\Psi_1=0$, $K_{BF}=\R$ and $s(u_0)\in\R$ for any $u_0\in K_{BF}$, such that $u=u_d$.

On the other hand, if $\bar b-\ka_2<0$, one knows that $s(y)$ has a global minimal $y^*$, according to Lemma \ref{lemma:1}. Note that $\Psi_0+\Psi_1y^*<0$ indicates $y^*\notin K_{BF}$, such that $s(u_0) \in \left[s\left(-\frac{\Psi_0}{\Psi_1}\right),+\infty\right]$ for any $u_0\in K_{BF}$ (note that the non-emptiness of $K_{BF}$ indicates $\Psi_0\geq 0$ if $\Psi_1=0$). Then, one can see that $s(u_0^*)=s\left(-\frac{\Psi_0}{\Psi_1}\right)$ if $\bar u_d\leq s\left(-\frac{\Psi_0}{\Psi_1}\right)$ and $s(u_0^*)=\bar u_d$ when $\bar u_d>s\left(-\frac{\Psi_0}{\Psi_1}\right)$, such that $s(u_0^*)=\max \left\{s\left(-\frac{\Psi_0}{\Psi_1}\right),\bar u_d\right\}$ and the closed-form solution to \eqref{cbfnlp} is   $u=h_{x_2}\max \left\{s\left(-\frac{\Psi_0}{\Psi_1}\right),\bar u_d\right\}$. Furthermore, $\Psi_0+\Psi_1 y^*\geq 0$ implies $y^*\in K_{BF}$, such that $s(u_0)\in [s(y^*), +\infty]$ for any $u_0\in K_{BF}$. Using the similar procedure shown above, one can conclude that the closed-form solution to \eqref{cbfnlp} is $u=h_{x_2}\max \left\{s(x^*),\bar u_d\right\}$.
\end{proof}

\subsection{Extension to More General Systems}\label{sec:maingeneral}
In this subsection, we will generalize the aCBF-based control design method proposed in Sec. \ref{sec:mainadaptive} to more general systems. Specifically, we will design a safe controller $u$ for the system \eqref{eqnsys2} with the same $f_\ta$, $\ta_i$, $\varphi_i$, $i\in[n]$, as those defined in \eqref{eqnstructure} and non-diagonal $g$ and $g_\la$
whose $(i,j)$-th entries can be expressed as 
\begin{equation}\label{ggla}
  (g)_{ij}=g_{ij}(x),\   (g_\la)_{ij}=\la^\top_{ij}\psi_{ij}(x),
\end{equation}
where $g_{ij}:\R^{m+n}\to\R$, $\psi_{ij}:\R^{m+n}\to\R^{q_{ij}}$ are known Lipschitz functions and $\la_{ij}\in\R^{q_{ij}}$ are vectors of unknown parameters, $i\in[n],j\in[n]$.

Similar to Assumption \ref{assump:1}, we assume that $\theta_i$ and $\la_{ij}$ are upper and lower bounded by known vectors. 
\begin{assumption}\label{assump:3}
For every $i,j\in[n]$,
there exist known vectors $\overline{\theta}_i, \underline{\theta}_i\in\mathbb{R}^{p_i}$ and $\overline{\la}_{ij},\underline{\la}_{ij}\in\mathbb{R}^{q_{ij}}$, such that  $\underline{\theta}_i\leq\theta_i \leq\overline{\theta}_{i}$ and  $\underline{\lambda}_{ij}\leq \la_{ij}\leq\overline{\lambda}_{ij}$.
\end{assumption}

Similar to \cite[Assumption 1]{xu2003robust}, we assume that $\tilde g\triangleq g+g_\la$ is away from the singularity point by letting the smallest singular value of $\frac{\tilde g(x)+\tilde g^\top (x)}{2}$ lower bounded by some known positive constant.

\begin{assumption}\label{assump:4}
Given functions $g(x),g_\la(x)$ in  the forms as shown in \eqref{ggla}, the matrix $\frac{\tilde g(x)+\tilde g^\top (x)}{2}$ is either uniformly positive definite or uniformly negative definite for all $x\in\mathscr{X}$,  where $\mathscr{X}\supset \C,\tilde g=g+g_\la$ and $\mathscr{X}\in\R^{m+n}$ is a compact set, i.e., there exists a positive constant $b^*>0$ such that $\underline{\sigma}\left(\frac{\tilde g(x)+\tilde g^\top (x)}{2}\right)\geq b^*,\forall x\in\mathscr{X}$,  
where $\underline{\sigma}(\cdot)$ represents the smallest singular value of a matrix. 
\end{assumption}

Without loss of generality, we assume that $\frac{\tilde g(x)+\tilde g^\top (x)}{2}$ is positive definite for any $x\in\C$ in this subsection. We select arbitrary values $\theta^0_i\in\R^{p_i}$, $\lambda^0_{ij}\in\R^{q_{ij}}$ satisfying $\underline{\theta}_i\leq\theta^0_i\leq \overline{\theta}_i$, $\underline{\lambda}_{ij}\leq\lambda^0_{ij}\leq \overline{\lambda}_{ij}$, $i,j\in[n]$, as the nominal values of $\theta_i$ and $\la_{ij}$, respectively.  We define
\begin{IEEEeqnarray}{rCl}
\IEEEyesnumber \label{matrixdefine}
\IEEEyessubnumber
\hspace{-5mm}\Theta &=& \begin{bmatrix}
\theta_1^\top &\theta_2^\top &\cdots&\theta_{n}^\top
\end{bmatrix}^\top,\; \Theta^0 = \begin{bmatrix}
\theta_1^{0\top} &\theta_2^{0\top} &\cdots&\theta_{n}^{0\top}
\end{bmatrix}^\top,\\
\IEEEyessubnumber
\hspace{-5mm}\Lambda &=& \begin{bmatrix}
\lambda_{11}^\top &\lambda_{12}^\top &\cdots&\lambda_{nn}^\top
\end{bmatrix}^\top,\; \Lambda^0 = \begin{bmatrix}
\lambda_{11}^{0\top} &\lambda_{12}^{0\top} &\cdots&\lambda_{nn}^{0\top}
\end{bmatrix}^\top, \\
\IEEEyessubnumber
\hspace{-5mm}\Omega_\varphi &=& \begin{bmatrix}
\varphi_1^\top &\varphi_2^\top &\cdots&\varphi_{n}^\top
\end{bmatrix}^\top, \;
\Omega_\psi = \begin{bmatrix}
\psi_{11}^\top &\psi_{12}^\top &\cdots&\psi_{nn}^\top
\end{bmatrix}^\top,\\
\hspace{-5mm}f_\theta^0 &=&\begin{bmatrix}
\theta_1^0\varphi_1\\\theta_2^0\varphi_2\\\vdots\\\theta_{n}^0\varphi_{n}
\end{bmatrix}, \;
\IEEEyessubnumber \label{matrixdefine:gla}
g_\la^0=\begin{bmatrix}
\la_{11}^{0\top} \psi_{11}&\la_{12}^{0\top} \psi_{12}&\cdots&\la_{1n}^{0\top} \psi_{1n}\\
\la_{21}^{0\top} \psi_{21}&\la_{22}^{0\top} \psi_{22}&\cdots&\la_{2n}^{0\top}  \psi_{2n}\\
\cdots &\cdots &\cdots&\cdots\\
\la_{n1}^{0\top} \psi_{n1}&\la_{n2}^{0\top}  \psi_{n2}&\cdots&\la_{nn}^{0\top} \psi_{nn}
    \end{bmatrix}.
\end{IEEEeqnarray}
and
\begin{equation}\label{munugeneral}
    \mu=\|\Theta-\Theta^0\|,\; \nu=\|\Lambda-\Lambda^0\|.
\end{equation}
According to Assumption \ref{assump:3}, one can see that
\begin{align*}
  \mu\leq \bar\mu\triangleq\sqrt{\sum_{j=1}^{M}\max\{(\overline\Theta_j-\Theta^0_j)^2,(\underline\Theta_j-\Theta^0_j)^2\}},\\
  \IEEEyessubnumber
  \nu\leq\bar \nu\triangleq\sqrt{\sum_{j=1}^{N}\max\{(\overline\Lambda_j-\Lambda^0_j)^2,(\underline\Lambda_j-\Lambda^0_j)^2\}},
\end{align*}
where
$M=\sum_{i=1}^{n}p_i$ and $N=\sum_{j=1}^{n}\sum_{i=1}^{n}q_{ij}$.
Analogous to Theorem \ref{theorem1}, the following theorem provides an aCBF-based controller that ensures the safety of system \eqref{eqnsys2} with $g$ and $g_\la$ defined in \eqref{ggla}.

\begin{theorem}\label{theorem3}
Consider the system \eqref{eqnsys2} with $f_\ta$ defined in \eqref{eqnstructure} and $g,g_\la$ defined in \eqref{ggla}, as well as the safe set $\C$ defined in \eqref{setc}. Suppose that\\
(i) Assumptions \ref{assump:0}, \ref{assump:3} and \ref{assump:4} hold;\\
(ii) There exist positive constants $\gamma,\epsilon_1,\epsilon_2,\gamma_\theta,
\gamma_{\lambda}>0$, such that the following set is non-empty:
\begin{equation}\label{cbfconditiongeneral1}
    K_{BF}^{g}\triangleq\{{\mathfrak u}\in\R\mid \Psi_0+\Psi_1 {\mathfrak u}\geq 0\}, \ \forall x\in\C,
\end{equation}
where
$\Psi_{0}=\mathscr{M}+h_{x_2} f_\theta^0 -(\ep_1+\ep_2)+\gamma\left(h-\frac{\bar\mu^2}{2\ga_\theta}-\frac{\bar\nu^2}{2\ga_\la}\right)$, $\Psi_{1}=h_{x_2}(g+g_\la^0)h_{x_2}^\top$, $h_{x_2}=\frac{\pa h}{\pa x_2}$, and $\mathscr{M}$ is the same as that defined in Theorem \ref{theorem1};\\
(iii) $\hat\mu$ and $\hat\nu$ are parameter estimations governed by the following adaptive laws:
\begin{IEEEeqnarray}{rCl}
\IEEEyesnumber \label{adaptive:bothgeneral}
\IEEEyessubnumber \label{adaptivelgeneralaw1}
\dot {\hat{\mu}} &=&-{\gamma} \hat\mu+\gamma_{\theta}\|h_{x_2}\|\|\Omega_\varphi\|, \\
\IEEEyessubnumber \label{adaptivelgeneralaw2}
\dot {\hat{\nu}} &=&-{\gamma} \hat\nu+\gamma_{\lambda}\|h_{x_2}\|^2|u_{0}|\|\Omega_\psi\|,
\end{IEEEeqnarray}
where $\hat\mu(0), \hat\nu(0)>0$ and $u_0\in\R$ is a Lipschitz function satisfying $u_0\in K_{BF}^g(x)$;\\
(iv) The following inequality holds: $h(x(0))\geq \frac{\hat\mu(0)^2+\bar{\mu}^2}{2\gamma_{\theta}}+\frac{\hat\nu(0)^2+\bar{\nu}^2}{2\gamma_{\lambda}}$.\\
Then, the control input $u=s_g(u_0)h_{x_2}^\top\in\R^n$ will make $h(x(t))\geq 0$ for any $t>0$, where
\begin{equation}\label{ugeneral}
    s_g(u_0)\triangleq u_0+\frac{\ka_{1,g}}{b^*}+\frac{\ka_{2,g}^2u_0^2}{b^*(\ka_{2,g}\|h_{ x_2}\||u_0|+\ep_2)},
\end{equation}
with $\ka_{1,g}=\frac{\hat\mu^2\|\Omega_\varphi\|^2}{\hat\mu \|\Omega_\varphi\|\left\|h_{x_2}\right\|+\epsilon_1}$ and $\ka_{2,g}=\hat\nu\|\Omega_\psi\|\|h_{x_2}\|$.
\end{theorem}

\begin{proof}
Assumption \ref{assump:4} indicates that, for any $v\in\R^{n}$, $v^\top (g+g_\la)v\geq b^*\|v\|^2$ \cite{xu2003robust}. Similar to the proof of Theorem \ref{theorem1}, one can see that $\hat\mu(t)\geq 0,\hat\nu(t)\geq 0,\forall t>0$. Define a candidate CBF $\bar h$ as $\bar h = h-\frac{1}{2\ga_\theta}\tilde{\mu}^2-\frac{1}{2\ga_\la}\tilde{\nu}^2$, 
where $\tilde\mu = \mu - \hat\mu$ and $\tilde\nu = \nu - \hat\nu$. 

We claim that $\dot{\bar{h}}\geq-\ga\bar h$ where $\dot{\bar{h}}$ is the time derivative of $\bar h$. Indeed, it is easy to see that $\dot{\bar{h}}\geq \mathscr{M} +h_{x_2}(f_\theta + (g+g_\la) u )+\frac{1}{\ga_\theta}\tilde\mu \dot{\hat{\mu}} +\frac{1}{\ga_\la}\tilde\nu \dot{\hat{\nu}}$. 
Substituting \eqref{ugeneral} into the inequality above yields $\dot{\bar{h}}\geq  \mathscr{M}+ h_{x_2} f_\theta +h_{x_2} (g+g_\la)h_{x_2}^\top u_0 + \frac{1}{\ga_\theta}\tilde\mu \dot{\hat{\mu}} +\frac{1}{\ga_\la}\tilde\nu \dot{\hat{\nu}}+
    h_{x_2}  (g+g_\la)h_{x_2}^\top \bigg(\frac{\hat\mu^2\|\Omega_\varphi\|^2}{b^*(\hat\mu \|\Omega_\varphi\|\|h_{x_2}\|+\epsilon_1)}+\frac{\hat\nu^2\|\Omega_\psi\|^2\|\|h_{x_2}\|^2u_{0}^2}{b^*(\hat\nu|u_{0}|\|\Omega_\psi\|\|\|h_{x_2}\|^2+\epsilon_2)}\bigg)
    \geq \mathscr{M}+h_{x_2}f^0_\theta +h_{x_2} (g+g^0_\la)h_{x_2}^\top u_0 + \frac{1}{\ga_\theta}\tilde\mu \dot{\hat{\mu}} +\frac{1}{\ga_\la}\tilde\nu \dot{\hat{\nu}}+h_{x_2} (f_\theta-f_\theta^0)
    +h_{x_2}(g_\la-g_\la^0)h_{x_2}^\top u_0\nonumber+\frac{\hat\mu^2\|h_{x_2}\|^2\|\Omega_\varphi\|^2}{\hat\mu \|\Omega_\varphi\|\|h_{x_2}\|+\epsilon_1}+\frac{\hat\nu^2\|\Omega_\psi\|^2\|\|h_{x_2}\|^4 u_{0}^2}{\hat\nu|u_{0}|\|\Omega_\psi\|\|\|h_{x_2}\|^2+\epsilon_2}
    \geq \Psi_0+\Psi_1u_0 +(\ep_1+\ep_2)-\ga\left(h-\frac{1}{2\ga_\theta}\bar{\mu}^2-\frac{1}{2\ga_\la}\bar{\nu}^2\right)+ \frac{1}{\ga_\theta}\tilde\mu \dot{\hat{\mu}} +\frac{1}{\ga_\la}\tilde\nu \dot{\hat{\nu}}-\mu\|h_{x_2}\|\|\Omega_\varphi\|
     -\nu\|h_{x_2}\|^2\|\Omega_\psi\||u_0|+\frac{\hat\mu^2\|h_{x_2}\|^2\|\Omega_\varphi\|^2}{\hat\mu \|\Omega_\varphi\|\|h_{x_2}\|+\epsilon_1}+\frac{\hat\nu^2\|\Omega_\psi\|^2\|h_{x_2}\|^4u_{0}^2}{\hat\nu|u_{0}|\|\Omega_\psi\|\|h_{x_2}\|^2+\epsilon_2}$, 
where the second inequality is from Assumption \ref{assump:4} and the third inequality comes from Lemma \ref{lemma1} shown in Appendix. Selecting $u_0\in K_{BF}^g$ we have $   \dot{\bar{h}}\geq \ep_1+\ep_2-\ga\left(h-\frac{1}{2\ga_\theta}\bar{\mu}^2-\frac{1}{2\ga_\la}\bar{\nu}^2\right)+\frac{1}{\ga_\theta}\tilde\mu\left(\dot{\hat{\mu}}-\ga_\theta\|h_{x_2}\|\|\Omega_\varphi\|\right)+\frac{1}{\ga_\la}\tilde\nu\left(\dot{\hat{\nu}}-\ga_\la\|h_{x_2}\|^2\|\Omega_\psi\||u_0|\right)
    -\hat\mu \|h_{x_2}\|\|\Omega_\varphi\|+\frac{\hat\mu^2\|h_{x_2}\|^2\|\Omega_\varphi\|^2}{\hat\mu \|\Omega_\varphi\|\|h_{x_2}\|+\epsilon_1}-\hat \nu\|h_{x_2}\|^2\|\Omega_\psi\||u_0| + \frac{\hat\nu^2\|\Omega_\psi\|^2\|h_{x_2}\|^4u_{0}^2}{\hat\nu|u_{0}|\|\Omega_\psi\|\|h_{x_2}\|^2+\epsilon_2}
    \geq 
    -\ga\left(h-\frac{1}{2\ga_\theta}\bar{\mu}^2-\frac{1}{2\ga_\la}\bar{\nu}^2\right)+\frac{1}{\ga_\theta}\tilde\mu\left(\dot{\hat{\mu}}-\ga_\theta\|h_{x_2}\|\|\Omega_\varphi\|\right)+\frac{1}{\ga_\la}\tilde\nu\left(\dot{\hat{\nu}}-\ga_\la\|h_{x_2}\|^2\|\Omega_\psi\||u_0|\right).$
Substituting  \eqref{adaptive:bothgeneral} into the inequality above, we have $\dot{\bar{h}}\geq -\ga\left(h-\frac{1}{2\ga_\theta}\bar{\mu}^2-\frac{1}{2\ga_\la}\bar{\nu}^2\right) -\frac{\ga}{\ga_\theta}\tilde\mu \hat\mu-\frac{\ga}{\ga_\la}\tilde\nu \hat\nu.$
Similar to the proof of Theorem \ref{theorem1}, one can see $\hat{\mu}\tilde{\mu}\leq \frac{\bar{\mu}^2}{2}-\frac{\tilde{\mu}^2}{2}$ and $\hat{\nu}\tilde{\nu}\leq \frac{\bar{\nu}^2}{2}-\frac{\tilde{\nu}^2}{2}$, which implies that $\dot{\bar{h}}\geq-\ga\bar h$. Our claim is thus proven. 
Note that $\bar h(x(0),0)\geq 0$ because of condition (iv). Hence, one can conclude that $\bar h(t)\geq 0,\forall t>0$, and thus, $h(x(t))\geq 0,\forall t>0$.
\end{proof}

\begin{remark}\label{remark:generaldisadvantage}
Compared with Theorem \ref{theorem1}, Theorem \ref{theorem3} provides a safety guarantee for a more general class of systems but the resulting safe controller tends to have more conservative performance. This is because the control $u\in \R^n$ is designed to have a particular structure $u=s_g(u_0)h_{x_2}^\top$, which requires $u$ always proportional to $h_{x_2}^\top$, to deal with the non-diagonal structures of $g$ and $g_\la$. How to improve the design to generate a less conservative safe controller will be our future work.
\end{remark}

The safe controller $u(x)$ in Theorem \ref{theorem3} can be obtained pointwise for any $x\in\C$ via solving the following optimization problem:
\begin{align}
\min_{u\in\R^n} \quad & \|u- u_d\|^2\label{cbfnlpgeneral}\\
\textrm{s.t.}\quad &  \Psi_0+\Psi_1 u_0\geq 0,\nonumber\\
\quad & u=s_g(u_0)h_{x_2}^\top,\nonumber
\end{align}
where $\Psi_0$ and $\Psi_1$ are defined in Theorem \ref{theorem3}. The closed-form solution of \eqref{cbfnlpgeneral} can be obtained by using Proposition \ref{theorem2}.

\section{Tightening Parameter Bounds via a Data-driven Approach}\label{sec:data}

The controller design proposed in Section \ref{sec:main} relies on the bounds of unknown parameters as shown in Assumptions \ref{assump:1} and \ref{assump:3}. If the prior knowledge of the parameter bounds is poor, the control performance tends to be conservative (see simulation examples in Section \ref{sec:simulation}).  In this section, we present a data-driven approach to get tighter bounds and more accurate nominal values for the unknown parameters. Combining the aCBF-based control design and the data-driven parameter tightening approach provides a mechanism to achieve safety with less conservatism.

Our data-driven method is inspired by the differential inclusion technique proposed in \cite{verginis2021safety}.  To better illustrate the main idea, we consider the system \eqref{eqnsys2} with $m=0$ and $n=1$ shown as follows: 
\begin{equation}\label{eqnsyssimple}
    \dot x = f(x)+f_u(x)+\theta^\top\varphi(x)+(g(x)+\la^\top \psi(x))u,
\end{equation}
where $x\in\R$ is the state, $u\in\R$ is the control input, $f:\R\to\R$ and $g:\R\to\R$ are known Lipschitz functions, $f_u:\R\to\R$ is an unknown (globally) Lipschitz function satisfying Assumption \ref{assump:0}, $\varphi:\R\to\R^p$, $\psi:\R\to\R^q$ are known functions, and $\theta\in\R^p,\la\in\R^q$ are unknown parameters.  The proposed method can be readily extended to systems with multiple inputs by considering each control channel separately. 

Recall that $x_i$ denotes the $i$-th entry of $x$ where $x$ is either a column or a row vector. Given a dataset $\mathscr{E}=\{x^i,\dot{x}^i,u^i\}_{i=1}^N$, the bounds of $\theta$, $\la$, and $f_u$ can be tightened as shown by the following theorem.

\begin{theorem}\label{theoremdatadriven}
Consider the system given in \eqref{eqnsyssimple}. Suppose that (i) Assumptions \ref{assump:0} and \ref{assump:1} hold; (ii) $f_{u}$ has a known Lipschitz constants $L$; (iii) a dataset $\mathscr{E}=\{x^i,\dot{x}^i,u^i\}_{i=1}^N$ is given. Define intervals  $\mcp^0=[\underline{\ta}, \overline{\ta}]^\top \in\mathbb{IR}^{1\times p}$ and $\Q^0=[\underline{\la}, \overline{\la}]^\top\in\mathbb{IR}^{1\times q}$. Let $x^0$ be an arbitrary state in $\C$ and define $\F^0=[ \underline{f}_{u}(x^0), \overline{f}_u(x^0)]$. For $i\in[N]$, $r\in[p]$, $s\in[q]$, define  
\begin{IEEEeqnarray}{rCl}
\IEEEyesnumber
\IEEEyessubnumber
\hspace{-5mm}\F^{i} &=& \left(
    \bigcap_{j=0}^{i-1}\{\F^{j}+L\|x^{i}-x^j\|[-1,1]\}
    \right)\bigcap\nonumber\\
    \IEEEyessubnumber
\hspace{-5mm}&& [\underline{f}_{u}(x^{i}), \overline{f}_u(x^{i})]\bigcap (y^{i}-\mcp^0\varphi^{i}-\Q^0\psi^iu^i),\label{boundfui}\\
\IEEEyessubnumber
\hspace{-5mm}v_0^i &=& (y^i-\F^{i}-\Q^0\psi^iu^i)\cap (\mcp^{i-1} \varphi^i),\\
\IEEEyessubnumber
\hspace{-5mm}v_{r}^i &=& (v_{r-1}^i-\mcp^{i-1}_{r}\varphi^i_{r})\cap \left(\sum_{l=r+1}^p\mcp^{i-1}_{l}\varphi^i_l\right),\\
\IEEEyessubnumber
\hspace{-5mm}\mcp^{i}_{r}\!&=&\!\begin{cases}
\!\!((v_{r-1}^i\!-\!\sum_{l=r+1}^p\!\!\mcp^{i-1}_{l}\!\varphi^i_l)\!\cap\!(\mcp^{i-1}_{r}\!\varphi^i_{r})))\!\frac{1}{\varphi^i_{r}}, \quad
&\!\!\!\!{\rm if} \ \varphi^i_r\neq 0,\\
\!\!\mcp^{i-1}_{r},  &\!\!\!\!{\rm otherwise},
\end{cases}
\end{IEEEeqnarray}
and
\begin{IEEEeqnarray}{rCl}
\IEEEyesnumber
\IEEEyessubnumber
\hspace{-3mm}w_0^i &=& (y^i-{\F}^i-\mcp^{N}\varphi^i)\cap (\Q^{i-1} \psi^iu^i),\\
\IEEEyessubnumber
\hspace{-3mm}w_{s}^i &=& (w_{s-1}^i-\Q^i_{s}\psi^i_s u^i)\cap\left(\sum_{l=s+1}^q\Q^{i-1}_{l}\psi^i_l u^i\right),\\
\IEEEyessubnumber
\hspace{-3mm}\Q^i_{s}&=&\begin{cases}
((w_{s-1}^i-\sum_{l=s+1}^q\Q^{i-1}_{l}\psi^i_l u^i)\cap\\
\indent\indent(\Q^{i-1}_{s}\psi^i_{s}u^i))\frac{1}{\psi^i_{s}u^i},
&{\rm if} \ \psi^i_{s}u^i\neq 0,\\
\Q^{i-1}_{s},  &{\rm otherwise},
\end{cases}
\end{IEEEeqnarray}
where $\varphi^i=\varphi(x^i)$, $\psi^i=\psi(x^i)$, and $y^i=\dot{x}^i-f(x^i)-g(x^i)u^i$. Then, $\ta^\top\in\mcp^N$, $\la^\top\in\Q^N$, and $f_u(x)\in\F(x)\triangleq\bigcap_{j=0}^N\{\F^j+L\|x-x^j\|[-1,1]\}$, for any $x\in\C$.
\end{theorem}
\begin{proof}
Note that $f_u(x^i)\in[\underline{f}_u(x^i),\overline{f}_u(x^i)]$ from Assumption \ref{assump:0} and $f_u(x^i)\in y^i-\mcp^0\varphi^i-\Q^0\psi^iu^i$ from
$f_u(x^i)=y^i-\ta^\top\varphi^i-\la^\top\psi^iu^i$, $\ta^\top\in\mcp^0$, and $\la^\top\in\Q^0$. One can see that $f_u(x^i)\in f_u(x^j)+L\|x^i-x^j\|[-1,1]$ holds for any $i\in[N]$ and $j=0,1,\cdots,i-1$ because $|f_u(x^i)-f_u(x^j)|\leq L\|x^i-x^j\|$ by the Lipschitzness of $f_u$. Hence, it is obvious that for any $i\in[N]$, $f_u(x^i)\in\left(\bigcap_{j=0}^{i-1}\{f_u(x^j)+L\|x^i-x^j\|[-1,1]\}\right)\bigcap[\underline{f}_{u}(x^{i}), \overline{f}_u(x^{i})]\bigcap\\ (y^{i}-\mcp^0\varphi^{i}-\Q^0\psi^iu^i)$ \cite{verginis2021safety}, which indicates $f_u(x^i)\in\F^i$ provided $f_u(x^k)\in\F^k$ for any $0\leq k<i$. Since $f_u(x^0)\in\F^0$,
using mathematical induction one can conclude that $f_u(x^i)\in\F^i$, $\forall i\in[N]$; thus, for any $x\in\C$, $f_u(x)\in\bigcap_{j=0}^N\{f_u(x^j)+L\|x-x^j\|[-1,1]\}\subset \F(x)$.

Next, we will prove that if $\ta^\top \in\mcp^{i-1}$, then  
$\ta^\top \varphi^i-\sum_{l=1}^r\ta_l\varphi^i_l\in v^i_r$ holds for any $0\leq r\leq p$. When $r=0$, one can see that $\ta^\top\varphi^i\in v^i_0$ since $\ta^\top\varphi^i=y^i-f_u(x^i)-\la^\top \psi^iu^i\in y^i-\F^i -\Q^0\psi^iu^i$ and $\ta^\top \varphi^i\in\mcp^{i-1}\varphi^i$. Then, we assume $\ta^\top \varphi^i-\sum_{l=1}^{r-1}\ta_l\varphi^i_l\in v^i_{r-1}$ holds. It can be seen that $\ta^\top \varphi^i-\sum_{l=1}^{r}\ta_l\varphi^i_l=\ta^\top \varphi^i-\sum_{l=1}^{r-1}\ta_l\varphi^i_l-\ta_r\varphi^i_r\in v^i_{r-1}-\ta_r\varphi^i_r\in v^i_{r-1}-\mcp^{i-1}_r\varphi^i_r$. On the other hand, one can see $\ta^\top \varphi^i-\sum_{l=1}^{r}\ta_l\varphi^i_l=\sum_{l=r+1}^p\ta_l\varphi^i_l\in \sum_{l=r+1}^p\mcp^{i-1}_l\varphi_l^i$. Summarizing the discussion above, one can conclude that $\ta^\top \varphi^i-\sum_{l=1}^r\ta_l\varphi^i_l\in  (v_{r-1}^i-\mcp^{i-1}_{r}\varphi^i_{r})\cap(\sum_{l=r+1}^p\mcp^{i-1}_{l}\varphi^i_l)=v_{r}^i$.

Finally, we will prove  $\ta^\top\in\mcp^i$ for any $0\leq i\leq N$ using mathematical induction. For $i=0$, $\ta^\top\in\mcp^0$ because of Assumption \ref{assump:1}. Then we assume $\ta^\top \in\mcp^{i-1}$. 
Note that $\ta^\top\varphi^i=y^i-f_u(x^i)-\la^\top \psi^iu^i\in y^i-\F^i -\Q^0\psi^iu^i$ and $\ta^\top \varphi^i\in\mcp^{i-1}\varphi^i$, which implies that $\ta^\top\varphi^i\in v_0^i$. It can be seen that for any $r\in[p]$ one has
$\ta_r\varphi^i_r=\ta^\top\varphi^i-\sum_{l=1}^{r-1}\ta_l\varphi_l^i-\sum_{l=r+1}^p\ta_l\varphi_l^i\in v^i_{r-1}-\sum_{l=r+1}^p \mcp^{i-1}_l\varphi_l^i$. Moreover, noticing that $\ta_r\varphi^i_r\in\mcp^{i-1}_r\varphi_r^i$, we have $\ta_r\in\mcp^i_r$ for any $r\in[p]$, which indicates $\ta^\top \in\mcp^i$. Following the similar procedure above, one can prove that $\la^\top \in\Q^i$. 
\end{proof}
\begin{remark}\label{remark:tightenbound}
With tighter bounds on $\ta$, $\la$ and $f_u$ provided by Theorem \ref{theoremdatadriven}, a larger admissible set $K_{BF}(x)$ as defined in \eqref{cbfcondition1} can be obtained. As a result, the data-driven-augmented aCBF-NLP controller tends to have a better control performance while always ensuring safety. It is expected that the system’s performance will be improved if the dataset $\mathscr{E}$ is large enough and the data in $\mathscr{E}$ are sufficiently ``diverse" (i.e., the whole state space is sufficient explored), but a formal proof is still under our investigation. 
The Lipschitz constant $L$ is needed in Theorem \ref{theoremdatadriven} to induce the bounds of $f_u$ from a finite number of data. A lot of existing work can be leveraged to estimating the Lipschitz constant of an unknown function, such as \cite{wood1996estimation,fazlyab2019efficient}. Moreover, the data-driven approach can be also combined with the aCBF-based controller shown in \eqref{cbfnlpgeneral} to reduce its conservatism. 
\end{remark}

\section{Simulation}
\label{sec:simulation}
In this section, three examples are provided to demonstrate the effectiveness of the proposed control method. More details about simulations  can be found at https://arxiv.org/abs/2302.08601.

\begin{example}\label{example1}
\label{singleinput}
Consider the following single-input system:
\begin{equation}
    \dot x=f_u+\theta_1\sin(x)+\theta_2  x^2+(\lambda_1+\lambda_2x^2)u,\label{sys:sim}
\end{equation}
where $x\in\R$ is the state and $u\in\R$ is the control input. The function $f_u=\cos(x)$ is unknown in the controller design; we choose the bounds of $f_u$ as $f_u\in[-2,2]$ such that Assumption \ref{assump:0} holds. The true values of the parameters $\theta_1=\theta_2=2,\lambda_1=1,\lambda_2=2$ are unknown in the controller design; we choose the bounds of these parameters as $\theta_1,\theta_2,\la_1,\la_2\in [-10,10]$ such that Assumption \ref{assump:1} holds. Note that loose bounds of the unknown parameters and the function are chosen deliberately. It is easy to verify that Assumption \ref{assump:2} is satisfied with $b=0.5$. We choose the safe set as $\C=\{x: h(x)\geq 0\}$ where $h(x)=x-1$, that is, we aim to make $x(t)\geq 1$ for all $t\geq 0$. The initial condition of \eqref{example1} is chosen as $x(0)=2$, the reference trajectory is selected  as $x_d=3\sin(t)$ and the 
nominal control $u_d$ is designed via feedback linearization. 

First, we demonstrate the performance of the safe controller obtained from  \eqref{cbfnlp1}.
The nominal values of the unknown parameters are $\theta_1^0 = \theta_2^0=\la_2^0 = 0,\la_1^0=0.5$, such that Condition (ii) of Theorem \ref{theorem1} holds because $\Psi_1=\la_1^0+\la_2^0x^2 = 0.5\neq 0$ for any $x\in\C$, which implies that $K_{BF}\neq \emptyset$. Other control parameters are selected such that Condition (iv) of Theorem \ref{theorem1} holds. Therefore, all conditions of Theorem \ref{theorem1} are satisfied. Applying the safe controller obtained from  \eqref{cbfnlp1}, the state evolution of the closed-loop system is shown as the blue line in Figure \ref{fig:1}. Then, we consider the aCBF-NLP-based safe controller combined with the data-driven approach.  We assume that $f_u$ has a Lipschitz constant $L=1$ and a dataset of 10 points is given. Applying the data-driven augmented, aCBF-NLP-based safe controller, the state evolution of the closed-loop system is shown as the pink line in Figure \ref{fig:1}.

From Fig. \ref{fig:1}, one can observe that the proposed aCBF-NLP controller, either with or without the data-driven technique, can ensure the safety of the system because the trajectory of $x$ always stays inside the safe region whose boundary is represented by the dashed red line, and the reference trajectory is well-tracked within the safe
set. Moreover, the performance of the data-driven augmented aCBF-NLP-based controller is less conservative since the tracking performance of the desired controller is better preserved inside the safe region and the state trajectory is allowed to approach the boundary of the safe set when the reference trajectory is outside the safe region.

\begin{figure}
    \centering
  \includegraphics[width=0.5\textwidth]{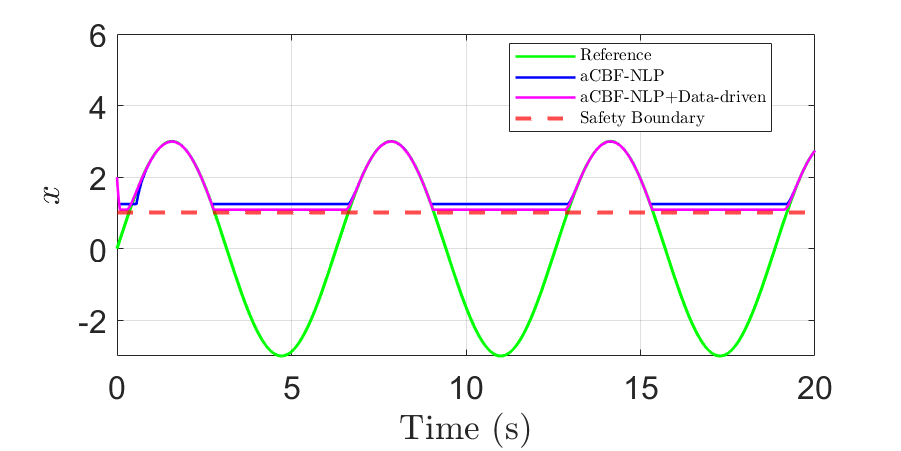}
    \caption{Evolution of the state variable $x$ of Example \ref{example1}. It can be seen that both the aCBF-NLP controller and the data-driven augmented aCBF-NLP controller can ensure safety as the trajectories of $x$ always stay in the safe region (i.e., above the dashed red line). One can also see that, when the data-driven technique developed in Theorem \ref{theoremdatadriven} is adopted, the aCBF-NLP controller has a better control performance. }
    \label{fig:1}
\end{figure}

\end{example}

\begin{example}\label{example2}
Consider the following adaptive cruise control system \cite{Xu2015ADHS}:
\begin{equation}\label{acceqn}
    \frac{\di}{\di t}\begin{bmatrix}
    D\\v_l\\v_f
    \end{bmatrix}=\begin{bmatrix}
    v_l-v_f\\a\\-\frac{1}{m}(f_0+f_1v_f+f_2v_f^2)
    \end{bmatrix}+\begin{bmatrix}
    0\\0\\\frac{1}{m}
    \end{bmatrix}u,
\end{equation}
where $v_l$ and $v_f$ are the velocities of the lead car and the following car, respectively, $D$ is the distance between the two vehicles, $u$ is the control input, 
$F_r\triangleq f_0+f_1v_f+f_2v_f^2$ is the aerodynamic drag term with constants $f_0,f_1,f_2$, and $m$ is the mass of the following car. The true values of the parameters $f_0=0.1 N,f_1=5 N\cdot s/m,f_2=0.25 N\cdot s^2/m,m=1650\;kg$ are unknown in the controller design. We assume that $f_0\in[0,10]$, $f_1\in[0,50]$, $f_2\in[0,20]$, $m\in[100,3000]$, and let $\theta=\frac{1}{m}[-f_0\ -f_1\ -f_2]^\top$ , $\varphi(v_f)=[1\ v_f\ v_f^2]^\top$, $\la=\frac{1}{m}$; one can easily see that Assumption  \ref{assump:1} is satisfied with $\theta_1\in[-0.1,0]$, $\theta_2\in[-0.5,0]$, $\theta_3\in[-0.2,0]$, $\la\in[0.00033,0.01]$. Note that in \eqref{acceqn} $f_u=0$, such that we selected $\underline{f}_u=\overline{f}_u=0$, from which one can see Assumption \ref{assump:0} is satisfied. Meanwhile, from \eqref{acceqn} one can easily verify that Assumption \ref{assump:2} holds with $b=1/3000$. The safety constraint of the following car is to keep a safe distance from the lead car, which can be expressed as $D/v_f\geq 1.8$ where 1.8 is the desired time headway in seconds. Therefore, the safe set is $\C=\{x: h(x)\geq 0\}$ where $h=D-1.8v_f$. The nominal controller $u_d$ is designed to keep the following car at a desired speed $v_{f, des}=22\; m/s$. We choose the nominal parameters $\ta^0=[-0.05\ -0.5\ -0.2]^\top$ and $\la^0=1/3000$, such that $\Psi_1=1.8^2/3000\neq 0$ for any $x\in\C$; thus,  $K_{BF}$ defined in \eqref{cbfcondition1} is non-empty, implying that Condition (ii) of Theorem \ref{theorem1} holds.

Applying the safe controller obtained from  \eqref{cbfnlp1}, the state and CBF evolution are shown as the blue lines in Figure \ref{fig:2}. Next, we consider the aCBF-NLP-based controller augmented with a dataset of 5 datapoints. The state and CBF evolution of the closed-loop system with the data-driven-augmented aCBF-NLP controller are shown as the brown lines in Figure \ref{fig:2}.
One can see that both controllers can ensure safety in the presence of parametric uncertainties since $h(t)\geq 0$ for any $t>0$, and the tracking performance is satisfactory when the reference trajectory is inside the safe region. Furthermore, the data-driven augmented aCBF-NLP controller has a slightly
better control performance in terms of maintaining the desired velocity because the bounds of the unknown parameters are tightened by the data-driven approach, as discussed in Remark \ref{remark:tightenbound}.

\begin{figure}[!htbp]
\centering
\begin{subfigure}{0.5\textwidth}
  \centering
  \includegraphics[width=1\textwidth]{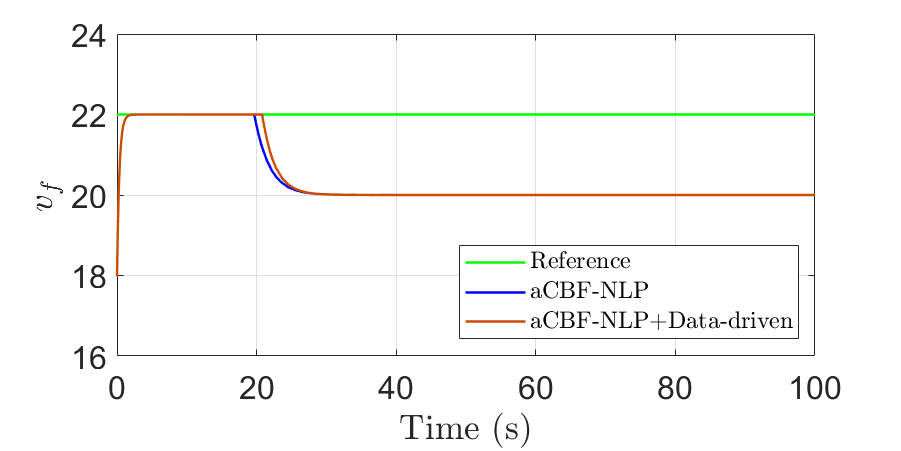}
  \caption{Evolution of $v_f$}
\end{subfigure}
\begin{subfigure}{0.5\textwidth}
  \centering
  \includegraphics[width=1\textwidth]{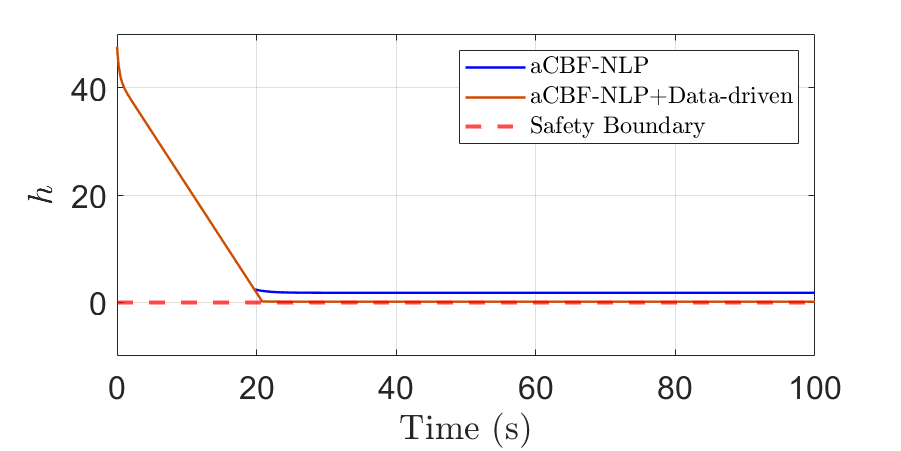}
  \caption{Evolution of  $h$}
\end{subfigure}
\caption{Simulation results of Example \ref{example2}. 
The aCBF-NLP controller, either with or without  the data-driven technique, can ensure safety of the system. When combined with the data-driven techniques, the aCBF-NLP controller has a slightly better control performance in terms of maintaining the desired velocity.
}
\label{fig:2}
\end{figure}
\end{example}

\begin{example}\label{example3}
\label{sec:simmass}
Consider the mass-spring system as follows:
\begin{equation}\label{example:masseqn}
    \frac{\di}{\di t}\begin{bmatrix}
    x_1\\x_2\\\dot x_1\\\dot x_2
    \end{bmatrix}=\begin{bmatrix}
    0&0&1&0\\0&0&0&1\\
    -\frac{k_1+k_2}{m_1}&\frac{k_2}{m_1}&0&0\\
    \frac{k_2}{m_2}&-\frac{k_2}{m_2}&0&0
    \end{bmatrix}\begin{bmatrix}
    x_1\\x_2\\\dot x_1\\\dot x_2
    \end{bmatrix}+\begin{bmatrix}
    0&0\\0&0\\
    \frac{1}{m_1}&0\\0&\frac{1}{m_2}
    \end{bmatrix}\begin{bmatrix}
    u_1\\u_2
    \end{bmatrix},
\end{equation}
where $x_1, x_2\in\R$ denote the position of two mass points, $u_1,u_2\in\R$ are control inputs, $m_1=m_2=0.2$ represent the mass, and $k_1=k_2=1$ denote the stiffness of two springs. We assume that all functions in \eqref{example:masseqn} are known, that is, $f_u=0$. 
Define $\theta_1=[-\frac{k_1+k_2}{m_1}\ \frac{k_2}{m_1}]^\top=[-10\ 5]^\top$, $\theta_2=[\frac{k_2}{m_2}\ -\frac{k_2}{m_2}]^\top=[5\ -5]^\top$, $\la_1=\frac{1}{m_1}=5$, and $\la_2=\frac{1}{m_2}=5$, which are unknown parameters in control design. One can easily verify that Assumption \ref{assump:1} is fulfilled
with $m_1,m_2\in[0.1,0.5]$ and $k_1,k_2\in[0,5]$, such that  $[-100\ 0]^\top\leq\theta_1\leq [0 \ 50]^\top$,
$[0\ -50]^\top\leq\theta_2\leq [50 \ 0]^\top$,  $\la_1\in[2,10]$, and $\la_2\in[2,10]$. It is obvious that Assumption \ref{assump:2} is fulfilled with $b_1=b_2=1$. The desired trajectories are selected as  $x_{1d}=0$, $x_{2d}=1+\sin (t)$, with a nominal PD controller $u_d$ designed to track $x_{1d}$, $x_{2d}$,
The safe set is defined as $\C=\{x: h(x)\geq 0\}$ with $h=x_2-x_1-0.5$, which aims to keep the distance between two masses. The initial conditions are selected as $x_1(0)=\dot x_1(0)=\dot x_2(0)=0$ and $x_2(0)=1$. Since the relative degree of $h$ is equal to 2, an exponential CBF that has a relative degree 1 is constructed as $\bar h=\dot h+15h=\dot x_2-\dot x_1+15(x_2-x_1-0.5)$. One can easily see that $\bar h\geq 0$ implies $h\geq 0$ because $\bar h(x_1(0),\dot x_1(0),x_2(0),\dot x_2(0))\geq 0$. Hence, we can use $\bar h$ to replace $h$ in Theorem \ref{theorem1} and \ref{theorem3}. 

We first consider the aCBF-NLP-based controller proposed in Theorem \ref{theorem1}. The nominal values of the unknown parameters are selected as $\ta_1^0=[-50\ 25]^\top$, $\ta_2^0=[25\ -25]^\top$, and $\la_1^0=\la_2^0=6$. Therefore, one can see that Condition (ii) of Theorem \ref{theorem1} holds (i.e.,$K_{BF}$ defined in \eqref{cbfcondition1} is non-empty) because $\Psi_1=[6\ 6]\neq 0$ for any $x\in\C$.
Applying the safe controller obtained from \eqref{cbfnlp1}, the state and CBF evolution are shown as the blue lines in Fig. \ref{fig:massresult}. Then, the aCBF-NLP-based controller is augmented with a dataset of 4 datapoints.
The state and CBF evolution are represented by the brown lines in Fig. \ref{fig:massresult}. It can be seen that both controllers can guarantee the safety since $h(x(t))\geq 0$ for any $t>0$, while the performance of the aCBF-NLP controller is improved if the data-driven approach is adopted.

Finally, we show how the results of Theorem \ref{theorem3} can be applied to \eqref{example:masseqn}. From now on we do not assume $g$ and $g_\la$ are diagonal matrices. It is easy to verify that Assumptions \ref{assump:0}, \ref{assump:3}, \ref{assump:4} hold true with $b^*=1$, $\la_{11},\la_{22}\in[-1,1]$, $\la_{12},\la_{21}\in [2,10]$, $[-100\ 0]^\top\leq\theta_1\leq [0 \ 50]^\top$, and
$[0\ -50]^\top\leq\theta_2\leq [50 \ 0]^\top$.
We select the nominal parameters  $\theta^0_1=[-50\ 25]^\top$, $\ta^0_2=[25\ -25]^\top$, $\la^0_{11}=\la^0_{22}=6$, $\la^0_{12}=\la^0_{21}=0$ and use the aforementioned exponential CBF $\bar h$. Thus, it is easy to verify that Condition (ii) of Theorem \ref{theorem3} holds (i.e.,$K_{BF}^g$ defined in \eqref{cbfconditiongeneral1} is non-empty) because $\Psi_1=12\neq 0$ for any $x\in\C$.

Applying the safe control law obtained  from \eqref{cbfnlpgeneral}, the state and CBF evolution are shown in Fig. \ref{fig:massresult2}, from which one can see that the safety is ensured since $h(x(t))\geq 0$ for any $t>0$. However, from Fig. \ref{fig:massresult2}(a) and \ref{fig:massresult2}(b), it can be seen that the control performance is conservative, i.e., the desired control performance is not well preserved inside the safe region. This phenomenon verifies what we discussed in Remark \ref{remark:generaldisadvantage}, i.e., $u$ might not be close to $u_d$ since it is always proportional to the partial derivative of $h$.

\begin{figure}[!htbp]
\centering
\begin{subfigure}{0.5\textwidth}
  \centering
  \includegraphics[width=1\textwidth]{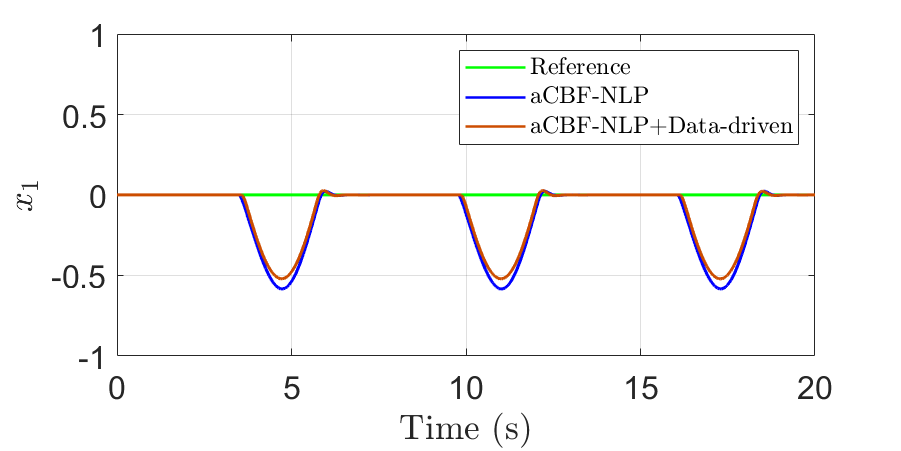}
  \caption{Evolution of $x_1$}
\end{subfigure}
\begin{subfigure}{0.5\textwidth}
  \centering
  \includegraphics[width=1\textwidth]{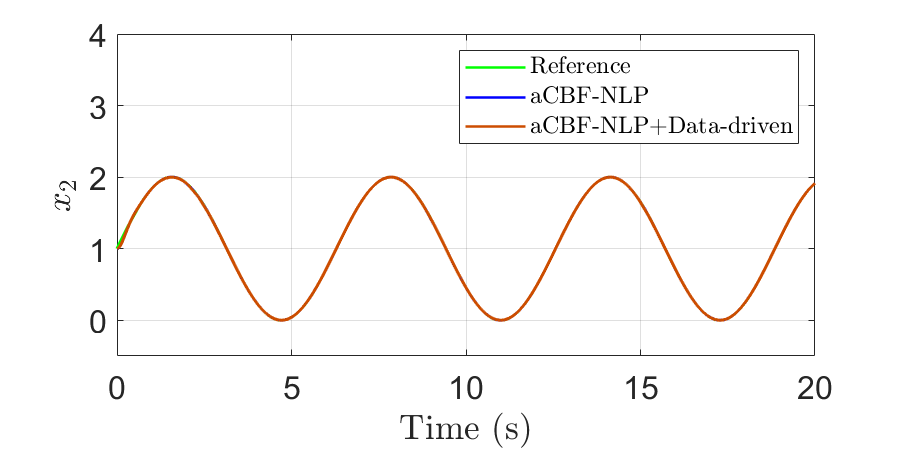}
  \caption{Evolution of $x_2$}
\end{subfigure}
\begin{subfigure}{0.5\textwidth}
  \centering
  \includegraphics[width=1\textwidth]{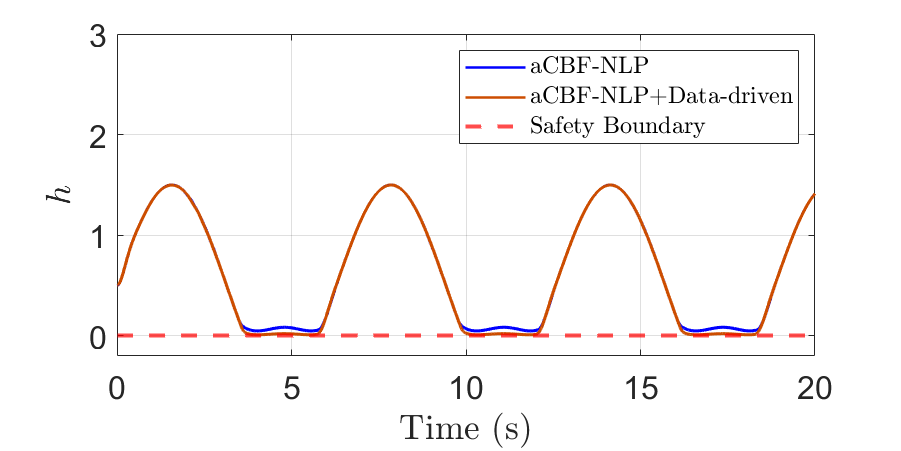}
  \caption{Evolution of  $h$}
\end{subfigure}
\caption{Simulation results of Example \ref{example3} using the control scheme shown in \eqref{cbfnlp1}. 
From (c) it can be seen that the proposed aCBF-NLP-based controller can guarantee safety as $h$ is always non-negative; from (a) it can be seen that, if the data-driven techniques are adopted, the control performance becomes less conservative since $x_1$ can track the reference trajectory better inside the safe region. 
}
\label{fig:massresult}
\end{figure}

\begin{figure}[!htbp]
\centering
\begin{subfigure}{0.5\textwidth}
  \centering
  \includegraphics[width=1\textwidth]{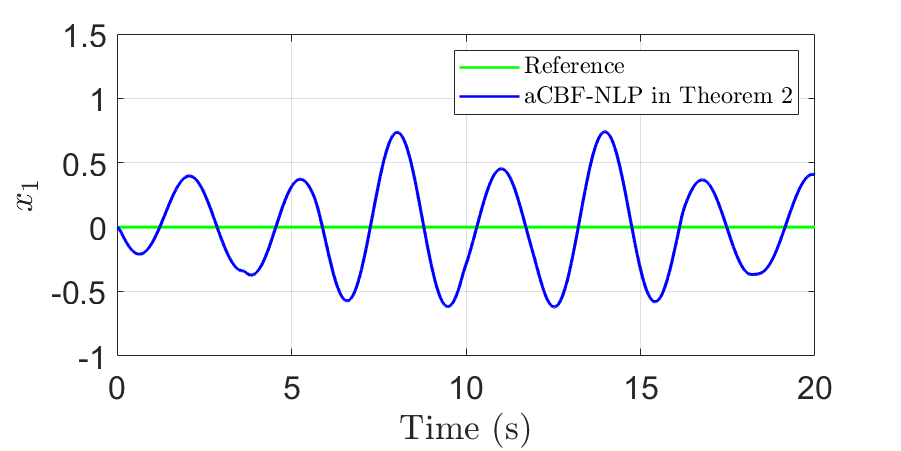}
  \caption{Evolution of $x_1$}
\end{subfigure}
\begin{subfigure}{0.5\textwidth}
  \centering
  \includegraphics[width=1\textwidth]{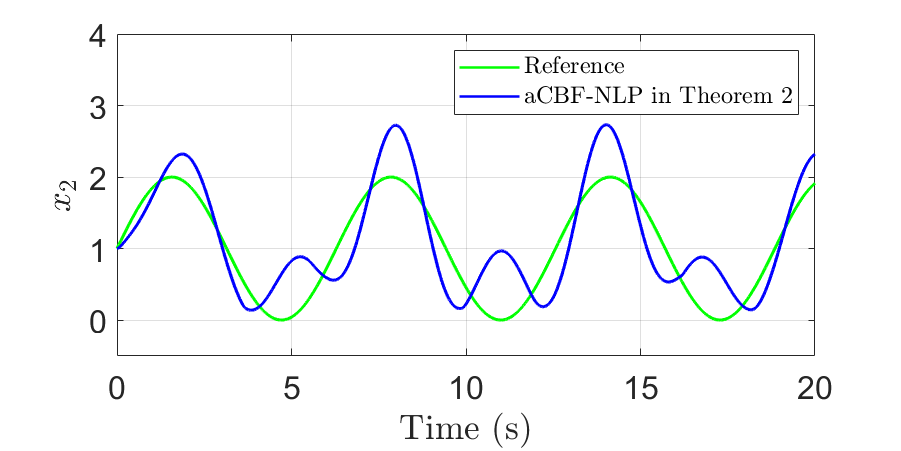}
  \caption{Evolution of $x_2$}
\end{subfigure}
\begin{subfigure}{0.5\textwidth}
  \centering
  \includegraphics[width=1\textwidth]{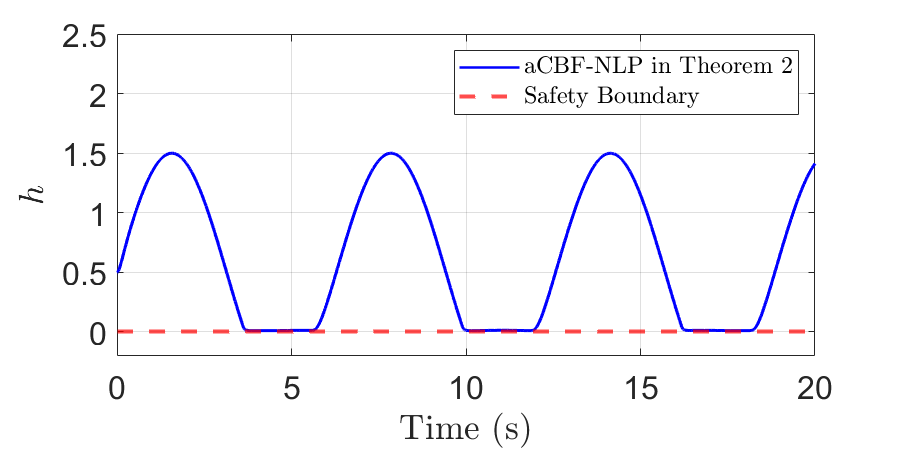}
  \caption{Evolution of  $h$}
\end{subfigure}
\caption{Simulation results of Example \ref{example3} using the control strategy shown in \eqref{cbfnlpgeneral}. From (c) it can be seen that the aCBF-NLP-based controller obtained by solving \eqref{cbfnlpgeneral} can guarantee safety; however, the control performance is unsatisfactory (i.e., the tracking performance of the desired controller is not well-preserved) due to the intrinsic conservatism discussed in Remark \ref{remark:generaldisadvantage}. 
}
\label{fig:massresult2}
\end{figure}
\end{example}

\section{Conclusion}
\label{sec:conclusion}
This paper proposes a singularity-free aCBF-NLP-based control strategy for systems with parametric uncertainties in both drift terms and control-input matrices, where the aCBF-NLP has a closed-form solution. Furthermore, a data-driven approach is developed to tighten the bounds of the unknown parameters and functions such that the performance of the proposed controller can be improved. Simulation results are also presented to validate the proposed  approach. Future work includes relaxing the assumptions of this paper and integrating this control method into learning-based control frameworks.

\section*{Appendices}
\setcounter{equation}{0}
\renewcommand\theequation{A.\arabic{equation}}

\begin{lemma}\label{lemma:1}
Define a function $s(\cdot)$ as
\begin{equation}
    s(y)=y+\frac{\ka_1}{b}+\frac{\ka_2^2 y^2}{b(\ka_2|h_{x_2}||y|+\ep_2)},
\end{equation}
where $\ka_1,\ka_2\geq 0$ and $b,\ep_2, |h_{x_2}|>0$ are considered as constants.
The function $s(\cdot)$  has the following properties:\\
(i) When  $\bar b-\kappa_2\geq 0$, where $\bar b=b|h_{x_2}|$, $s(y)$ is monotonically increasing with respect to $y$, and $\lim_{y\to-\infty}s(y)=-\infty$, $\lim_{y\to+\infty}s(y)=+\infty$;\\
(ii) When $\bar b-\kappa_2<0$, $s(y)$ has a global minimum 
$
y^*=\frac{\bar\ep_2[(\ka_2-\bar b)-\sqrt{\ka_2(\ka_2-\bar b)}]}{\ka_2(\ka_2-\bar b)},
$
where $\bar\ep_2=\ep_2/|h_{x_2}|$, and $\lim_{y\to-\infty}s(y)=\lim_{y\to+\infty}s(y)=+\infty$.
\end{lemma}
\begin{proof}
Note that
the derivative of $s(y)$ with respect to $y$
can be expressed as
\begin{equation}
    \frac{\di s}{\di y}=\begin{cases}
    1+ \frac{\ka_2^2y(\ka_2 y+2\bar\ep_2)}{\bar b(\ka_2 y+\bar\ep_2)^2}, & {\rm if} \  y\geq 0,\\
    \frac{(\bar b-\ka_2)\ka_2^2y^2-2\ka_2\bar\ep_2(\bar b-\ka_2)y+\bar b\bar\ep_2^2}{\bar b(-\ka_2 y+\bar\ep_2)^2}, & {\rm if} \ y<0.
    \end{cases}
\end{equation}

(i) If $\ka_2=0$, then $s(y)=\frac{\ka_1}{b}+y$, from which one can see that the statement is true. We assume $\ka_2\neq 0$ in the following analysis. 
It can be seen that if $\bar b-\ka_2\geq 0$, $\frac{\di s}{\di y}> 0$ for any $y\in \R$, such that $s(y)$ is monotonically increasing. Meanwhile, one can see  $\lim_{y\to+\infty}s(y)=+\infty$ and $\lim_{y\to-\infty}s(y)=\lim_{y\to-\infty}\frac{y\bar b(-\ka_2y+\bar\ep_2)+\ka_2^2y^2}{\bar b(-\ka_2y+\bar\ep_2)}+\frac{\ka_1}{b}\xlongequal[]{z=-y} \lim_{z\to+\infty}\frac{-\bar b(\ka_2z^2+\bar\ep_2z)+\ka_2^2z^2}{\bar b(\ka_2z+\bar\ep_2)}+\frac{\ka_1}{b}$. Define $\alpha(z)=-\bar b(\ka_2z^2+\bar\ep_2z)+\ka_2^2z^2$, $\beta(z)=\bar b(\ka_2z+\bar\ep_2)$, and $I=(0, +\infty)$, from which one can see $\lim_{y\to-\infty}s(y)=\frac{\ka_1}{b}+\lim_{z\to+\infty}\frac{\alpha(z)}{\beta(z)}$. Since $\alpha$, $\beta$ are differentiable with respect to $z$ and $\beta'(z)\neq 0$ for any $z\in I$, applying L'Hôpital's rule gives
$\lim_{y\to-\infty}s(y)= \lim_{z\to+\infty}\frac{-2(\bar b-\ka_2)\ka_2z-\bar b\bar\ep_2}{\bar b\ka_2}+\frac{\ka_1}{b}=-\infty$.

(ii) If $\bar b-\ka_2<0$, $\frac{\di s}{\di y}> 0$ still holds when $y\geq 0$. When $y< 0$, asking $\frac{\di s}{\di y}=0$ gives a stationary point $y^*$. It can be verified $\frac{\di^2 s}{\di y^2}\bigg|_{y=y^*}=\frac{2(\ka_2(\ka_2-\bar b))^{\frac{3}{2}}}{\ka_2\bar b\bar\ep_2}>0$, such that $y^*$ is a global minimum. Similarly, one can prove $\lim_{y\to+\infty}s(y)=\infty$ and $   \lim_{y\to-\infty} s(y)\xlongequal[]{z=-y} \lim_{z\to+\infty}\frac{\ka_2(\ka_2-\bar b)z^2-\bar b\bar \ep_2 z}{\bar b(\bar\ep_2+\ka_2 z)}+\frac{\ka_1}{b}=\lim_{z\to+\infty}\frac{2\ka_2(\ka_2-\bar b)z-\bar b\bar\ep_2}{\bar b \ka_2}+\frac{\ka_1}{b}=+\infty,$
where the second equality arises from L'Hôpital's rule (the conditions of L'Hôpital's rule can be verified using the similar procedure in (i)).
\end{proof}

\begin{lemma}\label{lemma1}
For any $a\in\R^{n}$, $b\in\R^{n}$, $c\in\R$, the following inequalities hold:
\begin{IEEEeqnarray}{rCl}
\IEEEyesnumber \label{matrixcauchy}
\IEEEyessubnumber \label{matrixcauchy:1}
a^\top (f_\theta-f_\theta^0) &\geq&  -\mu\|a\| \|\Omega_\varphi\|,\\
\IEEEyessubnumber\label{matrixcauchy:2}
b^\top (g_\la-g_\la^0) b c &\geq& -\nu \|\Omega_\psi\|\|b\|^2|c|,
\end{IEEEeqnarray}
where $f_\theta$ is defined in \eqref{eqnstructure},
$g_\la$ is defined in \eqref{ggla}, $f_\ta^0$, $g_\la^0$, $\Omega_\varphi$, $\Omega_\psi$ are defined in \eqref{matrixdefine}, and $\mu$, $\nu$ are defined in \eqref{munugeneral}.
\end{lemma}
\begin{proof}
One can verify that $a^\top (f_\theta-f_\theta^0)\geq -\|a\|\|f_\theta-f_\theta^0\|=-\|a\|\sqrt{\sum_{i=1}^{n}((\theta_i-\ta_i^0)^\top\varphi_i)^2}\geq -\|a\|\sqrt{\sum_{i=1}^{n}\|\ta_i-\ta_i^0\|^2\|\varphi_i\|^2}
\geq -\|a\|\sqrt{\sum_{i=1}^{n}\|\ta_i-\ta_i^0\|^2}\\\sqrt{\sum_{i=1}^{n}\|\varphi_i\|^2}=-\|a\|\|\Theta-\Theta^0\|\|\Omega_\varphi\|=-\mu\|a\|\|\Omega_\varphi\|$, 
where the first and second inequality are derived from Cauchy-Schwarz inequality and the third inequality comes from the fact $\sum_{k=1}^{n} x_k^2y_k^2 \leq \left(\sum_{k=1}^{n} x_k^2 \right)
    \left(\sum_{k=1}^{n} y_k^2 \right), \ \forall x_k,y_k\in\R.$ Therefore, \eqref{matrixcauchy:1} holds.

Similarly, using Cauchy-Schwarz inequality, one can get $b^\top (g_\la-g_\la^0)bc\geq -|b^\top (g_\la-g_\la^0) b||c|\geq -\ \|g_\la-g_\la^0\|\|b\|^2|c|.$ 
Invoking the definition of the Frobenius norm, $\|g_\la-g_\la^0\|$ satisfies $\|g_\la-g_\la^0\|=\sqrt{\sum_{i=1}^{n}\sum_{j=1}^{n}((\la_{ij}-\la_{ij}^0)^\top \psi_{ij})^2}\leq \sqrt{\sum_{i=1}^{n}\sum_{j=1}^{n}\|\la_{ij}-\la_{ij}^0\|^2 \|\psi_{ij}\|^2}
    \leq
\sqrt{\sum_{i=1}^{n}\sum_{j=1}^{n}\|\psi_{ij}\|^2}\\
\sqrt{\sum_{i=1}^{n}\sum_{j=1}^{n}\|\la_{ij}-\la_{ij}^0\|^2}=\|\Lambda-\Lambda^0\|\|\Omega_\psi\|=\nu\|\Omega_\psi\|$. Therefore, \eqref{matrixcauchy:2} holds.

\end{proof}

\bibliographystyle{elsarticle-num}
\bibliography{adaptiveCBF.bib}

\begin{thebibliography}{10}
\expandafter\ifx\csname url\endcsname\relax
  \def\url#1{\texttt{#1}}\fi
\expandafter\ifx\csname urlprefix\endcsname\relax\def\urlprefix{URL }\fi
\expandafter\ifx\csname href\endcsname\relax
  \def\href#1#2{#2} \def\path#1{#1}\fi

\bibitem{Xu2015ADHS}
X.~Xu, P.~Tabuada, A.~Ames, J.~Grizzle, Robustness of control barrier functions
  for safety critical control, IFAC-PapersOnLine 48~(27) (2015) 54--61.

\bibitem{ames2016control}
A.~D. Ames, X.~Xu, J.~W. Grizzle, P.~Tabuada, Control barrier function based
  quadratic programs for safety critical systems, IEEE Trans. Autom. Control
  62~(8) (2016) 3861--3876.

\bibitem{garg2021robust}
K.~Garg, D.~Panagou, Robust control barrier and control {L}yapunov functions
  with fixed-time convergence guarantees, in: 2021 American Control Conference
  (ACC), 2021, pp. 2292--2297.

\bibitem{nguyen2021robust}
Q.~Nguyen, K.~Sreenath, Robust safety-critical control for dynamic robotics,
  IEEE Trans. Autom. Control 67~(3) (2021) 1073--1088.

\bibitem{verginis2021safety}
C.~K. Verginis, F.~Djeumou, U.~Topcu, Learning-based, safety-constrained
  control from scarce data via reciprocal barriers, in: 2021 IEEE 60th
  Conference on Decision and Control (CDC), 2021, pp. 83--89.

\bibitem{buch2021robust}
J.~Buch, S.-C. Liao, P.~Seiler, Robust control barrier functions with
  sector-bounded uncertainties, IEEE Control Syst. Lett. 6 (2021) 1994--1999.

\bibitem{wang2022disturbance}
Y.~Wang, X.~Xu, Disturbance observer-based robust control barrier functions,
  in: 2022 American Control Conference (ACC), 2023, pp. 3681--3687.

\bibitem{narendra1989stable}
K.~S. Narendra, A.~M. Annaswamy, Stable Adaptive Systems, Prentice-Hall, 1989.

\bibitem{aastrom1995adaptive}
K.~J. {\AA}str{\"o}m, B.~Wittenmark, Adaptive Control (2nd Ed), Addison-Wesley,
  1995.

\bibitem{KKK95}
M.~Krsti{\'c}, P.~V. Kokotovi{\'c}, I.~Kanellakopoulos, Nonlinear and
  {A}daptive {C}ontrol {D}esign, John Wiley \& Sons, Inc., 1995.

\bibitem{ioannou1996robust}
P.~A. Ioannou, J.~Sun, Robust {A}daptive {C}ontrol, Prentice-Hall, 1996.

\bibitem{astolfi2008nonlinear}
A.~Astolfi, D.~Karagiannis, R.~Ortega, Nonlinear and Adaptive Control with
  Applications, Springer-Verlag, 2008.

\bibitem{sastry1989adaptive}
S.~S. Sastry, A.~Isidori, Adaptive control of linearizable systems, IEEE Trans.
  Autom. Control 34~(11) (1989) 1123--1131.

\bibitem{kanellakopoulos1991systematic}
I.~Kanellakopoulos, P.~V. Kokotovi{\'c}, A.~S. Morse, Systematic design of
  adaptive controllers for feedback linearizable systems, IEEE Trans. Autom.
  Control 36 (1991) 1241 -- 1253.

\bibitem{KKK92}
M.~Krsti{\'c}, I.~Kanellakopoulos, P.~Kokotovi{\'c}, Adaptive nonlinear control
  without overparametrization, Syst. \& Control Lett. 19~(3) (1992) 177--185.

\bibitem{anderson1986stability}
B.~D. Anderson, R.~R. Bitmead, C.~R. Johnson~Jr, P.~V. Kokotovic, R.~L. Kosut,
  I.~M. Mareels, L.~Praly, B.~D. Riedle, Stability of Adaptive Systems:
  Passivity and Averaging Analysis, MIT Press, 1986.

\bibitem{kosut1987stability}
R.~Kosut, B.~Anderson, I.~Mareels, Stability theory for adaptive systems:
  Method of averaging and persistency of excitation, IEEE Trans. Autom. Control
  32~(1) (1987) 26--34.

\bibitem{tao2014multivariable}
G.~Tao, Multivariable adaptive control: A survey, Automatica 50~(11) (2014)
  2737--2764.

\bibitem{annaswamy2021historical}
A.~M. Annaswamy, A.~L. Fradkov, A historical perspective of adaptive control
  and learning, Annu. Rev. Control 52 (2021) 18--41.

\bibitem{KK95}
M.~Krsti{\'c}, P.~V. Kokotovi{\'c}, Control {L}yapunov functions for adaptive
  nonlinear stabilization, Syst. \& Control Lett. 26~(1) (1995) 17--23.

\bibitem{taylor2020adaptive}
A.~J. Taylor, A.~D. Ames, Adaptive safety with control barrier functions, in:
  2020 American Control Conference (ACC), 2020, pp. 1399--1405.

\bibitem{lopez2020robust}
B.~T. Lopez, J.-J.~E. Slotine, J.~P. How, Robust adaptive control barrier
  functions: An adaptive and data-driven approach to safety, IEEE Control Syst.
  Lett. 5~(3) (2020) 1031--1036.

\bibitem{zhao2020adaptive}
P.~Zhao, Y.~Mao, C.~Tao, N.~Hovakimyan, X.~Wang, Adaptive robust quadratic
  programs using control {L}yapunov and barrier functions, in: 2020 IEEE 59th
  Conference on Decision and Control (CDC), 2020, pp. 3353--3358.

\bibitem{black2021fixed}
M.~Black, E.~Arabi, D.~Panagou, A fixed-time stable adaptation law for
  safety-critical control under parametric uncertainty, in: 2021 European
  Control Conference (ECC), 2021, pp. 1328--1333.

\bibitem{isaly2021adaptive}
A.~Isaly, O.~S. Patil, R.~G. Sanfelice, W.~E. Dixon, Adaptive safety with
  multiple barrier functions using integral concurrent learning, in: 2021
  American Control Conference (ACC), 2021, pp. 3719--3724.

\bibitem{cohen2022high}
M.~H. Cohen, C.~Belta, High order robust adaptive control barrier functions and
  exponentially stabilizing adaptive control {L}yapunov functions, in: 2022
  American Control Conference (ACC), 2022, pp. 2233--2238.

\bibitem{wang2022observer}
Y.~Wang, X.~Xu, Observer-based control barrier functions for safety critical
  systems, in: 2022 American Control Conference (ACC), 2022, pp. 709--714.

\bibitem{wang2022robust}
S.~Wang, B.~Lyu, S.~Wen, K.~Shi, S.~Zhu, T.~Huang, Robust adaptive
  safety-critical control for unknown systems with finite-time elementwise
  parameter estimation, IEEE Trans. Syst. Man Cybern.: Syst. (2022).

\bibitem{huang2022safety}
C.~Huang, L.~Long, Safety-critical model reference adaptive control of switched
  nonlinear systems with unsafe subsystems: A state-dependent switching
  approach, IEEE Trans. Syst. Man Cybern.: Syst. (2022).

\bibitem{azimi2021exponential}
V.~Azimi, S.~Hutchinson, Exponential control {L}yapunov-barrier function using
  a filtering-based concurrent learning adaptive approach, IEEE Trans. Autom.
  Control 67~(10) (2022) 5376--5383.

\bibitem{verginis2022funnel}
C.~K. Verginis, Funnel control for uncertain nonlinear systems via zeroing
  control barrier functions, IEEE Control Syst. Lett. 7 (2022) 853--858.

\bibitem{bechlioulis2008robust}
C.~P. Bechlioulis, G.~A. Rovithakis, Robust adaptive control of feedback
  linearizable {MIMO} nonlinear systems with prescribed performance, IEEE
  Trans. Autom. Control 53~(9) (2008) 2090--2099.

\bibitem{xu2003robust}
H.~Xu, P.~A. Ioannou, Robust adaptive control for a class of {MIMO} nonlinear
  systems with guaranteed error bounds, IEEE Trans. Autom. Control 48~(5)
  (2003) 728--742.

\bibitem{horn2012matrix}
R.~A. Horn, C.~R. Johnson, Matrix Analysis, Cambridge {U}niversity {P}ress,
  2012.

\bibitem{moore2009introduction}
R.~E. Moore, R.~B. Kearfott, M.~J. Cloud, Introduction to {I}nterval
  {A}nalysis, SIAM, 2009.

\bibitem{isidori1985nonlinear}
A.~Isidori, Nonlinear {C}ontrol {S}ystems: {A}n {I}ntroduction, Springer, 1985.

\bibitem{xu2018constrained}
X.~Xu, Constrained control of input--output linearizable systems using control
  sharing barrier functions, Automatica 87 (2018) 195--201.

\bibitem{khalil2002nonlinear}
H.~K. Khalil, Nonlinear {S}ystems, Prentice-Hall, 2002.

\bibitem{nguyen2016exponential}
Q.~Nguyen, K.~Sreenath, Exponential control barrier functions for enforcing
  high relative-degree safety-critical constraints, in: 2016 American Control
  Conference (ACC), 2016, pp. 322--328.

\bibitem{tan2021high}
X.~Tan, W.~S. Cortez, D.~V. Dimarogonas, High-order barrier functions:
  Robustness, safety, and performance-critical control, IEEE Trans. Autom.
  Control 67~(6) (2021) 3021--3028.

\bibitem{wood1996estimation}
G.~Wood, B.~Zhang, Estimation of the {L}ipschitz constant of a function, J.
  Global Optim. 8 (1996) 91--103.

\bibitem{fazlyab2019efficient}
M.~Fazlyab, A.~Robey, H.~Hassani, M.~Morari, G.~Pappas, Efficient and accurate
  estimation of {L}ipschitz constants for deep neural networks, Advances in
  Neural Information Processing Systems 32 (2019).

\end{thebibliography}

\end{sloppypar}
\end{document}